\documentclass[aps,superscriptaddress,11pt,nofootinbib]{revtex4-2}

\usepackage{xcolor}
\usepackage{amsmath,mathtools}
\usepackage{enumerate}
\usepackage{graphicx,epsfig,amsmath,amssymb,verbatim,color}
\usepackage{dsfont}
\usepackage{float}
\usepackage{tikz}
\usepackage{mathrsfs}   
\usepackage[normalem]{ulem}   

\usepackage{hyperref}
\hypersetup{colorlinks=true,%
citecolor=blue,%
linkcolor=blue,%
filecolor=blue,%
urlcolor=blue,%
breaklinks=true}
\usepackage{cleveref}
\usepackage{graphicx}
\usepackage{xcolor}
\usepackage{url}
\newtheorem{definition}{Definition}
\newtheorem{proposition}{Proposition}
\newtheorem{lemma}[proposition]{Lemma}
\newtheorem{theorem}[proposition]{Theorem}
\newtheorem{remark}{Remark}
\newtheorem{corollary}[proposition]{Corollary}
\newenvironment{proof}{\noindent \textbf{{Proof.}~}}{\hfill $\blacksquare$}

\newcommand{\nc}{\newcommand}
\nc{\bra}[1]{\langle#1|}
\nc{\ket}[1]{|#1\rangle}
\nc{\ketbra}[2]{|#1\rangle\!\langle#2|}
\nc{\braket}[2]{\langle#1|#2\rangle}
\nc{\proj}[1]{| #1\rangle\!\langle #1 |}
\nc{\tr}{\operatorname{Tr}}
\nc{\ox}{\otimes}
\nc{\1}{{\mathds{1}}}
\nc{\supp}{{\operatorname{supp}}}

\renewcommand{\dim}{\operatorname{dim}}

\nc{\cA}{{\cal A}}
\nc{\cB}{{\cal B}}
\nc{\cC}{{\cal C}}
\nc{\cD}{{\cal D}}
\nc{\cE}{{\cal E}}
\nc{\cF}{{\cal F}}
\nc{\cG}{{\cal G}}
\nc{\cH}{{\cal H}}
\nc{\cI}{{\cal I}}
\nc{\cJ}{{\cal J}}
\nc{\cK}{{\cal K}}
\nc{\cL}{{\cal L}}
\nc{\cM}{{\cal M}}
\nc{\cN}{{\cal N}}
\nc{\cO}{{\cal O}}
\nc{\cP}{{\cal P}}
\nc{\cQ}{{\cal Q}}
\nc{\cR}{{\cal R}}
\nc{\cS}{{\cal S}}
\nc{\cT}{{\cal T}}
\nc{\cV}{{\cal V}}
\nc{\cX}{{\cal X}}
\nc{\cY}{{\cal Y}}
\nc{\cZ}{{\cal Z}}
\nc{\cW}{{\cal W}}

\def\a{\alpha}
\nc{\id}{{\operatorname{id}}}

\allowdisplaybreaks

\begin{document}
\title{\Large Quantifying the unextendibility of entanglement\footnote{A preliminary version of this paper 
has been published in the Proceedings of the 2020 IEEE International Symposium on Information
Theory (ISIT 2020)~\cite{wang2020quantification}.}}

\author{Kun Wang}
\affiliation{Institute for Quantum Computing, Baidu Research, Beijing 100193, China}

\author{Xin Wang}
\affiliation{Thrust of Artificial Intelligence, Information Hub, 
The Hong Kong University of Science and Technology (Guangzhou), Guangzhou, China}
\affiliation{Institute for Quantum Computing, Baidu Research, Beijing 100193, China}

\author{Mark M. Wilde}
\affiliation{Hearne Institute for Theoretical Physics, Department of Physics and Astronomy,
Center for Computation and Technology, Louisiana State University, Baton Rouge, Louisiana 70803, USA}
\affiliation{School of Electrical and Computer Engineering, Cornell University, Ithaca, New
York 14850, USA}

\begin{abstract}
Entanglement is a striking feature of quantum mechanics, and it has a key property called unextendibility. In this paper, we present a framework for quantifying and investigating the unextendibility of general bipartite quantum states. First, we define the unextendible entanglement, a family of entanglement measures based on the concept of a state-dependent set of free states. The intuition behind these measures is that the more entangled a bipartite state is, the less entangled each of its individual systems is  with a third party. Second, we demonstrate that the unextendible entanglement is an entanglement monotone under two-extendible quantum operations, including local operations and one-way classical communication as a special case. Normalization and faithfulness are two other desirable properties of unextendible entanglement, which we establish here. We further show that the unextendible entanglement provides efficiently computable benchmarks for the rate of exact entanglement or secret key distillation, as well as the overhead of probabilistic entanglement or secret key distillation.

\end{abstract}

\date{\today}
\maketitle

\tableofcontents

\section{Introduction}

Quantum entanglement is a striking phenomenon of quantum mechanics and a key ingredient in many quantum
information processing tasks \cite{Horodecki2009,Hayashi2017b,Wilde2017book}.
Unshareability or unextendibility  is one of the key features of entanglement \cite{W89a,DPS02,DPS04}. Unextendibility asserts that the more entangled a bipartite state is, the less entangled its individual systems can be  with a third party. Given that classical correlations can be shared among many parties, it emphasizes a key distinction between classical and quantum correlations. 
The ``monogamy of quantum entanglement'' \cite{Terhal2004}, which has been described and studied in many ways \cite{Coffman2000,Terhal2004,Koashi2004,Yang2006,Osborne2006,Adesso2007b,Lancien2016,Gour2017a}, is also connected to this special characteristic of unextendibility. 

Unextendibility of entanglement plays a key role in frontier reaseach on Bell inequalities \cite{Terhal2003}, quantum key distribution \cite{Moroder2006a,Myhr2009a,Khatri2017}, and quantum communication \cite{Nowakowski2009,Kaur2018,BBFS18}. Some recent works~\cite{Kaur2018,KDWW21} developed a systematic resource theory of unextendibility to investigate the unextendibility of bipartite quantum states.
The quantifiers of unextendibility defined in \cite{Kaur2018,KDWW21} generalize several known measures introduced in the literature \cite{Nowakowski2009,B08,Moroder2006a,HMW13}. In Remark~\ref{rem:compl-approaches-unext} below, we comment on key differences between the approach of \cite{Kaur2018,KDWW21} and our approach taken here.

In this work, we present a resource theory of unextendible entanglement by taking into account a state-dependent collection of free states, which is motivated by the fundamental importance of unextendibility in understanding quantum entanglement. In particular, our contributions are as follows:
\begin{enumerate}

\item We first introduce the unextendible entanglement, which is a class of entanglement measures that minimizes the generalized
divergence between a given state $\rho_{AB}$ and any possible reduced state $\rho_{AB'}$ of an extension $\rho_{ABB'}$ of $\rho_{AB}$ (Section~\ref{sec:quant-unext}). This method is based on the idea that a bipartite state's separate systems are less entangled with a third party, the more entangled the state is.

\item 
Second, we show that our measures of unextendible entanglement are monotone under two-extendible operations. This also means that they are monotone under local operations and one-way classical communication (1-LOCC) (Section~\ref{sec:generalized divergence}). We also prove faithfulness and normalization of unextendible entanglement, two more desirable characteristics of the family of entanglement measures. We additionally take into account a variety of important R\'enyi relative entropy-based unextendible entanglement examples (Section~\ref{sec:alpha unext}), some of which are efficiently computable (Section~\ref{sec:efficiently-comp-unext}).

\item Third, we show how the unextendible entanglement can be utilized to effectively establish fundamental limitations on both the overhead and rate of secret key and entanglement distillation (Section~\ref{sec:applications}). These most recent findings suggest that the unextendible entanglement provides a useful lens into the virtues of entanglement as a resource.
\end{enumerate}

\begin{remark}[Comparison of our results with \cite{Kaur2018,KDWW21}]
\label{rem:compl-approaches-unext}
Before we delve into the rest of the paper, let us  briefly remark on the similarities and differences between our present paper and \cite{Kaur2018,KDWW21}. The main aspect in which these works are similar is that they both establish ways of quantifying entanglement in terms of the notion of unextendibility. However, beyond this conceptual link, the works are completely different in their approaches and the resulting applications.

The papers \cite{Kaur2018,KDWW21} establish measures of unextendibility in terms of a comparison between the state of interest and the set of $k$-extendible states, with the measure being called the relative entropy of unextendibility. The consequences of this choice are that the relative entropy of unextendibility does not give bounds on the asymptotic one-way distillable entanglement or the asymptotic quantum capacity, mainly due to \cite[Lemma~10]{KDWW21}, which states that the relative entropy of $k$-unextendibility can never be larger than $\log_2 k$, regardless of the system size. Thus, this quantity is only applicable for providing upper bounds on the rates of these tasks in the non-asymptotic regime.

In contrast, as stated above, here we compare the bipartite state $\rho_{AB}$ of interest, on systems $AB$, to another bipartite state $\sigma_{AB'}$ on systems $AB'$, such that there exists a tripartite state on $ABB'$ with marginals given by $\rho_{AB}$ and $\sigma_{AB'}$. Thus, this measure, called unextendible entanglement here, is completely different from the relative entropy of unextendibility from \cite{Kaur2018,KDWW21}. Furthermore, as we show in our paper, the unextendible entanglement is applicable in the asymptotic regime, giving upper bounds on the rate of probabilistic distillation of entanglement and secret key.

Thus, the approach here and the approach taken in \cite{Kaur2018,KDWW21} are complementary approaches for quantifying the unextendibility of entanglement.
\end{remark}
 
\section{Preliminaries}
\label{sec:preliminaries}

We begin in this section by establishing some notation and reviewing some definitions.

\subsection{Notations}

Throughout, quantum
systems are denoted by $A$, $B$, and $C$, have finite dimensions, and are associated with their underlying Hilbert spaces.
Systems described by the same letter are assumed
to be isomorphic: $A_1 \cong A_2$ and $A\cong A'$. The set of linear operators acting on system~$A$ is denoted by $\cL(A)$, the set of
Hermitian operators acting on system~$A$ is denoted by $\operatorname{Herm}(A)$, the set of positive semi-definite operators by
$\cP(A)$, and the set of quantum states (density operators) by $\cS(A)$. For a linear operator $L$, we use $L^T$ to denote its transpose and $L^\dagger$
to denote its conjugate transpose.
The $\alpha$-norm of $L$ is defined as $\left\Vert L\right\Vert_\alpha\coloneqq (\tr[(\sqrt{L^\dagger L})^\alpha])^{1/\alpha}$.
Especially, the trace norm is defined as $\lVert L\rVert_1 \coloneqq  \tr\sqrt{L^\dagger L}$.
The identity operator is denoted by $\1_A$.
An operator of the form $L_A\otimes\1_B$ is written simply $L_A$.
Quantum states of system $A$ are denoted by $\rho_A\in\cS(A)$ and pure (rank-one) quantum states by $\Psi_A$.
A maximally entangled state $\Phi^d_{AB}$ of Schmidt rank~$d$ is
defined as $\Phi^d_{AB} \coloneqq
\frac{1}{d}\sum_{i,j=1}^d \vert i \rangle\!\langle j \vert_A \ox \vert i \rangle\!\langle j \vert_B$,
where $\{\ket{i}_A\}_i$ and $\{\ket{i}_B\}_i$ are orthonormal bases for systems~$A$ and $B$, respectively.
A maximally entangled state $\Phi^2_{AB}$ is also known as an \emph{ebit}.
A purification of a quantum state $\rho_A\in\cS(A)$ is a pure state
$\Psi_{AR}^\rho\in\cS(AR)$ such that $\tr_R\Psi_{AR}^\rho = \rho_A$,
where $R$ is called the purifying system.
An extension of a quantum state $\rho_A\in\cS(A)$ is a quantum state (not necessarily pure)
$\rho_{AE}\in\cS(AE)$ such that $\tr_E\rho_{AE} = \rho_A$,
where $E$ is called the extension system.
Quantum channels are completely positive (CP) and trace preserving (TP) maps
from $\cL(A)$ to $\cL(B)$ and denoted by $\cN_{A\to B}\in\mathcal{Q}(A\to B)$.
The identity channel from $\cL(A)$ to $\cL(A)$ is denoted by~$\id_A$.

\subsection{Quantum divergences}

Let $\mathbb{R}$ denote the field of real numbers.
A functional $\mathbf{D}:\cS(A)\times\cS(A)\to\mathbb{R}\cup\{+\infty\}$ is a generalized
divergence~\cite{polyanskiy2010arimoto} if for arbitrary Hilbert spaces $\cH_A$ and $\cH_B$, arbitrary states
$\rho_A,\sigma_A\in\cS(A)$, and an arbitrary channel $\cN_{A\to B}\in\cQ(A\to B)$, the following data-processing inequality holds
\begin{align}
      \mathbf{D}(\rho_A\|\sigma_A)
\geq  \mathbf{D}(\cN_{A\to B}(\rho_A)\|\cN_{A\to B}(\sigma_A)).
\label{eq:DP-ineq-GD-def}
\end{align}

The data-processing inequality above implies that there is a constant $c \in \mathbb{R}$ such that $\mathbf{D}(\rho_A\|\sigma_A) \geq c$ for all states $\rho$ and $\sigma$. Indeed, we can choose the channel $\cN_{A\to B}$ in \eqref{eq:DP-ineq-GD-def} to be the trace and replace channel $\rho \to \operatorname{Tr}[\rho]\omega$, where $\omega$ is a state. Then applying \eqref{eq:DP-ineq-GD-def}, we conclude that $\mathbf{D}(\rho_A\|\sigma_A) \geq \mathbf{D}(\omega\|\omega)$. Furthermore, considering the trace and replace channel $\rho \to \operatorname{Tr}[\rho]\tau$, where $\tau$ is a state, and applying the data-processing inequality twice in opposite directions implies that $\mathbf{D}(\omega\|\omega) = \mathbf{D}(\tau\|\tau)$ for all states $\omega$ and $\tau$. Thus, we can set $c = \mathbf{D}(\omega\|\omega)$, justifying the claim, and we see that
$\mathbf{D}(\rho_A\|\sigma_A)$
takes its minimal value when the two states are equal, i.e., $\rho = \sigma$. Without loss of generality, we can then redefine the generalized divergence to be $\mathbf{D}(\rho_A\|\sigma_A) - c$, so that the redefined quantity satisfies $\mathbf{D}(\omega\|\omega) =0 $ for every state $\omega$. We make this assumption in what follows.

We then call a generalized divergence $\mathbf{D}$ faithful if $\mathbf{D}(\rho_A\|\sigma_A) =0$ if and only if $\rho_A=\sigma_A$. 
Examples of interest are the quantum relative entropy \cite{Ume62}, the
Petz-R\'{e}nyi relative entropies~\cite{petz1986quasi}, the sandwiched R\'{e}nyi
relative entropies~\cite{muller2013quantum,wilde2014strong}, and the geometric R\'{e}nyi
relative entropies~\cite{matsumoto2015,Fang2019a}.
Let us recall the definitions of these relative  entropies.
The Petz--R\'enyi relative entropy is defined for $\alpha\in(0,1)\cup(1,\infty)$ as follows:
\begin{align}
D_{\alpha}(\omega\Vert\tau) &  \coloneqq \frac{1}{\alpha-1}\log Q_{\alpha}(\omega\Vert\tau),\\
Q_{\alpha}(\omega\Vert\tau) &  \coloneqq \tr[ \omega^{\alpha}%
\tau^{  1-\alpha}],\label{eq:quasi-Petz}
\end{align}
where logarithms are taken to base two throughout this paper and
$D_1(\omega\Vert\tau)$ is defined as the limit of $D_{\alpha}(\omega\Vert\tau)$ for $\alpha\to1$.
The sandwiched R\'enyi relative entropy is defined for $\alpha\in(0,1)\cup(1,\infty)$ as
\begin{align}
\widetilde{D}_{\alpha}(\omega\Vert\tau) &  \coloneqq \frac{1}{\alpha-1}\log\widetilde{Q}_{\alpha}(\omega\Vert\tau),\label{eq:sandwiched-Renyi}\\
\widetilde{Q}_{\alpha}(\omega\Vert\tau) 
&\coloneqq  \left\Vert \omega^{1/2} \tau^{\left(1-\alpha\right)/2\alpha}\right\Vert_{2\alpha}^{2\alpha} = \tr\!\left[\left(\tau^{\left(1-\alpha\right)/2\alpha}\omega
            \tau^{\left(1-\alpha\right)/2\alpha}\right)^{\alpha}\right].
\end{align}
Moreover, $\widetilde{D}_1(\omega\Vert\tau)$ is defined
as the limit of $\widetilde{D}_{\alpha}(\omega\Vert\tau)$ for $\alpha\to1$.
The geometric R\'enyi relative entropy is defined for
$\alpha\in(0,1)\cup(1,\infty)$ as follows:
\begin{align}
\widehat{D}_{\alpha}(\omega\Vert\tau) &  \coloneqq \frac{1}{\alpha-1}\log\widehat{Q}_{\alpha}(\omega\Vert\tau),\\
\widehat{Q}_{\alpha}(\omega\Vert\tau) &  \coloneqq \tr\!\left[ \tau
\left(\tau^{-1/2}\omega\tau^{-1/2}\right)^{\alpha}\right], \label{eq:quasi-geometric}
\end{align}
and $\widehat{D}_1(\omega\Vert\tau)$ is defined
as the limit of $\widehat{D}_{\alpha}(\omega\Vert\tau)$ for $\alpha\to1$.
For all of the above quantities,
 if $\alpha >1$ and $\operatorname{supp}(\omega)\not\subseteq \operatorname{supp}(\tau)$, then
\begin{equation}
D_{\alpha}(\omega\Vert\tau)
=
\widetilde{D}_{\alpha}(\omega\Vert\tau)
=
\widehat{D}_{\alpha}(\omega\Vert\tau)
=+\infty.
\end{equation}
The Petz--R\'enyi relative entropy obeys the data-processing inequality for $\alpha \in [0,2]$ \cite{petz1986quasi},  the sandwiched R\'enyi relative entropy obeys the data-processing inequality for
$\alpha
\in[1/2,\infty]$ \cite{FL13,Bei13} (see also \cite{W18opt}), and the geometric R\'enyi relative entropy obeys the data-processing inequality for $\alpha \in (0,2]$ \cite{matsumoto2015,HM17}.
In the case that $\alpha=1$, we take the limit $\alpha \to 1$ to find that \cite{petz1986quasi,muller2013quantum,wilde2014strong}
\begin{equation}
D(\omega\Vert\tau) \coloneqq  \lim_{\alpha \to 1}
\widetilde{D}_{\alpha}(\omega\Vert\tau) = \lim_{\alpha \to 1}
D_{\alpha}(\omega\Vert\tau)
=\operatorname{Tr}[\omega\left[  \log\omega-\log\tau\right]  ],
\label{eq:rel-ent}
\end{equation}
where $D(\omega\Vert\tau)$ is the well-known quantum relative entropy \cite{Ume62}.
The geometric R\'enyi relative entropy converges to the Belavkin--Staszewski relative entropy \cite{BS82}
in the limit $\alpha \to 1$:
\begin{equation}
\lim_{\alpha \to 1}\widehat{D}_{\alpha}(\omega\Vert \tau)=
\widehat{D}(\omega\Vert \tau) \coloneqq  \tr[\omega\log (\omega^{1/2} \tau^{-1} \omega^{1/2})],
\end{equation}
which is known to obey the data-processing inequality  \cite{HM17}.
In the case that $\alpha=\infty$, we take the limit $\alpha \to \infty$ to find that \cite{muller2013quantum}
\begin{equation}
D_{\max}(\omega\Vert\tau)
\coloneqq  \lim_{\alpha \to \infty}\widetilde{D}_{\alpha}(\omega\Vert\tau)
=\inf\left\{  \lambda\in\mathbb{R}:\omega\leq2^{\lambda}\tau\right\}  ,
\label{eq:max-rel-ent-limit}
\end{equation}
where $D_{\max}(\omega\Vert\tau)$ is the max-relative entropy \cite{Datta2009}. See \cite{KW20book} for a detailed review of all of these quantities.

The following faithfulness property ensures that the above divergences are non-negative for normalized states
and vanishes only if both arguments are equal (proof in Appendix~\ref{appx:faithfulness-renyi-rel-ents}). 
It is useful for later analysis.

\begin{lemma}[Faithfulness]
\label{lem:faithfulness-renyi-rel-ents}
Let $\alpha \in (0,\infty)$, and let $\omega$ and $\tau$ be quantum states.
Then $D_{\alpha}(\omega\Vert\tau) \geq 0$ and $D_{\alpha}(\omega\Vert\tau) = 0$ if and only if $\omega = \tau$.
The same is true for the sandwiched and geometric R\'enyi relative entropies.
\end{lemma}

\section{Quantifying the unextendibility of entanglement}
\label{sec:quant-unext}

This section is structured as follows.
In Section~\ref{sec:state-dependent-free-state},
we define the set of free states determined by a given bipartite quantum state $\rho_{AB}$.
In Section~\ref{sec:generalized divergence}, we introduce the generalized unextendible entanglement.
In Section~\ref{sec:alpha unext}, we consider special cases of the generalized unextendible entanglement,
which include those based on the quantum relative entropy \cite{Ume62},
the Petz--R\'enyi relative entropy \cite{petz1986quasi},
the sandwiched R\'enyi relative entropy \cite{muller2013quantum,wilde2014strong},
and the geometric R\'enyi relative entropy \cite{matsumoto2015} (see also \cite{Fang2019a}).
We establish several desirable properties for these measures,
including the fact that they are monotone under two-extendible operations.

\subsection{State-dependent set of free states}
\label{sec:state-dependent-free-state}

Unlike the usual framework of quantum resource theories~\cite{Chitambar2018} and that which was established for unextendibility in
\cite{Kaur2018}, the free states in our resource theory are \textit{state-dependent}, which is inspired by the definition and meaning of unextendibility of entanglement.
Note that state-dependent resource theories
were previously considered in a different context \cite{Regula2019b}. 
To be specific, given a bipartite state $\rho_{AB}$,
the free states are those bipartite states that are possibly shareable between Alice and a third party $B'$, where the system
$B'$ is isomorphic to $B$. Mathematically, the set of free states corresponding to a state $\rho_{AB}$ is defined as follows: 
\begin{align}
\label{eq:free-states}
\cF_{\rho_{AB}} \coloneqq  \{\sigma_{AB'}\, : \rho_{AB} = \tr_{B'}[ \omega_{ABB'}], \ \sigma_{AB'}
= \tr_B [\omega_{ABB'}],\ \omega_{ABB'}\in\cS(ABB')\}.
\end{align}
Furthermore, if $\rho_{AB}\in\cF_{\rho_{AB}}$, then $\rho_{AB}$ is \textit{two-extendible},
indicating that there exists an extension~$\omega_{ABB'}$ of $\rho_{AB}$ such that
$\rho_{AB} =  \tr_{B'}[ \omega_{ABB'}]= \tr_B [\omega_{ABB'}]$.
Separable states are always two-extendible.
Interestingly, there are entangled states that are also two-extendible;
a concrete example can be found in Eq.~(7) of~\cite{doherty2014entanglement}.
As so, the resource theory under consideration differs from commonly studied 
entanglement quantifiers where separable states are usually regarded as constituting the full set of free states.

We introduce \emph{selective two-extendible operations} as the free operations under consideration,
which generalize the two-extendible channels previously defined and investigated in~\cite{Kaur2018}.
Mathematically, a selective two-extendible operation consists of a set of CP\ maps of the following form:%
\begin{equation}
\left\{  \mathcal{E}_{AB\rightarrow A^{\prime}B^{\prime}}^{y}\right\}  _{y\in \mathcal{Y}},
\end{equation}
such that $\sum_{y\in \mathcal{Y}}\mathcal{E}_{AB\rightarrow A^{\prime}B^{\prime}}^{y}$ is trace-preserving,
each map $\mathcal{E}_{AB\rightarrow A^{\prime}B^{\prime}}^{y}$ is
completely positive and two-extendible, in the sense that there exists
an extension map $\mathcal{E}_{AB_{1}B_{2}\rightarrow A^{\prime}B_{1}^{\prime}B_{2}^{\prime}}^{y}$
satisfying~\cite{Kaur2018}
\begin{enumerate}
\item Channel extension:
\begin{align}
\forall\rho_{AB_1B_2}\in\cS(AB_1B_2):
\mathcal{E}_{AB\rightarrow A^{\prime}B^{\prime}}^{y}(\rho_{AB_1})
= \tr_{B'_2}\!\left[\mathcal{E}_{AB_{1}B_{2}\rightarrow A^{\prime}B_{1}^{\prime
}B_{2}^{\prime}}^{y}(\rho_{AB_1B_2})\right],
\label{eq:marginal}
\end{align}
where $\rho_{AB_1} = \tr_{B_2}\rho_{AB_1B_2}$ is the marginal quantum state of $\rho_{AB_1B_2}$.
\item Permutation covariance:%
\begin{equation}
\mathcal{E}_{AB_{1}B_{2}\rightarrow A^{\prime}B_{1}^{\prime}B_{2}^{\prime}%
}^{y}\circ\mathcal{W}_{B_{1}B_{2}}=\mathcal{W}_{B_{1}^{\prime}%
B_{2}^{\prime}}\circ\mathcal{E}_{AB_{1}B_{2}\rightarrow A^{\prime}%
B_{1}^{\prime}B_{2}^{\prime}}^{y},
\label{eq:perm-cov}
\end{equation}
where $\mathcal{W}_{B_1B_2}$ is the unitary swap channel.
\end{enumerate}
The two conditions above ensure that the two marginal operations%
\begin{align}
\mathcal{E}_{AB_{1}\rightarrow A^{\prime}B_{1}^{\prime}}^{y}(\rho_{AB_1})  &
\coloneqq  \tr_{B^\prime_2}\!\left[\mathcal{E}_{AB_{1}B_{2}\rightarrow A^{\prime}B_{1}^{\prime
}B_{2}^{\prime}}^{y}(\rho_{AB_1B_2})\right],\\
\mathcal{E}_{AB_{2}\rightarrow A^{\prime}B_{2}^{\prime}}^{y}(\rho_{AB_2})  &
\coloneqq  \tr_{B^\prime_1}\!\left[\mathcal{E}_{AB_{1}B_{2}\rightarrow A^{\prime}B_{1}^{\prime
}B_{2}^{\prime}}^{y}(\rho_{AB_1B_2})\right],
\end{align}
where $\rho_{AB_1} = \tr_{B_2}\rho_{AB_1B_2}$ is the marginal quantum state,
are operations that are in fact equal to the original operation:%
\begin{equation}
\mathcal{E}_{AB\rightarrow A^{\prime}B^{\prime}}^{y}=\mathcal{E}%
_{AB_{1}\rightarrow A^{\prime}B_{1}^{\prime}}^{y}=\mathcal{E}_{AB_{2}%
\rightarrow A^{\prime}B_{2}^{\prime}}^{y}.
\end{equation}
We emphasize that the channel extension defined in~\eqref{eq:marginal} 
must have well defined channel marginals, which is a nontrivial property
and commonly referred to as the quantum channel marginal problem.
We refer interested readers to~\cite{hsieh2022quantum}
for comprehensive research on this problem.

A two-extendible channel \cite{Kaur2018} is a special case of
a selective two-extendible operation for which~$|\mathcal{Y}|=1$.

Now we introduce selective 1-LOCC operations as a special case of selective two-extendible operations.
Mathematically, a selective 1-LOCC operation consists of a set of CP maps of the following form:%
\begin{equation}
\left\{  \mathcal{L}_{AB\rightarrow A^{\prime}B^{\prime}}^{y}\equiv\sum
_{x \in \mathcal{X}}\mathcal{F}_{A\rightarrow A^{\prime}}^{x,y}\otimes\mathcal{G}%
_{B\rightarrow B^{\prime}}^{x,y}\right\}  _{y\in \mathcal{Y}}, \label{eq:selective-1W-LOCC}%
\end{equation}
where $\cX$ and $\cY$ are the alphabets,
each map $\mathcal{F}_{A\rightarrow A^{\prime}}^{x,y}$ is completely
positive, the sum map $\sum_{x\in \mathcal{X},y\in \mathcal{Y}}\mathcal{F}_{A\rightarrow A^{\prime}}^{x,y}$
is trace-preserving, and each map $\mathcal{G}_{B\rightarrow B^{\prime}}%
^{x,y}$ is completely positive and trace-preserving.
A 1-LOCC channel is a special case of a selective 1-LOCC operation for which~$| \mathcal{Y}| = 1$.
The fact that selective 1-LOCC operations are a special case of selective
two-extendible operations can be verified by constructing the extension map of
$\mathcal{L}_{AB\rightarrow A^{\prime}B^{\prime}}^{y}$ in
\eqref{eq:selective-1W-LOCC}\ in the following way:%
\begin{equation}
\mathcal{L}_{AB_{1}B_{2}\rightarrow A^{\prime}B_{1}^{\prime}B_{2}^{\prime}%
}^{y}\coloneqq \sum_{x}\mathcal{F}_{A\rightarrow A^{\prime}}^{x,y}\otimes
\mathcal{G}_{B_{1}\rightarrow B_{1}^{\prime}}^{x,y}\otimes\mathcal{G}%
_{B_{2}\rightarrow B_{2}^{\prime}}^{x,y}.
\end{equation}
Such a choice satisfies the channel extension and permutation covariance
properties required of selective two-extendible operations.
We emphasize that similar inclusions were observed in \cite{Kaur2018} for 
two-extendible channels and 1-LOCC channels.

\subsection{Generalized unextendible entanglement}
\label{sec:generalized divergence}

We introduce a family of entanglement measures called \textit{unextendible entanglement}.
Given a bipartite state $\rho_{AB}$, the key idea behind these measures is
to minimize the divergence between $\rho_{AB}$ and any possible reduced state $\rho_{AB'}$ 
of an extension $\rho_{ABB'}$
of $\rho_{AB}$. These measures are intuitively motivated by the fact that the 
more that a bipartite state is entangled, 
the less that each of its individual systems is entangled with a third party.

\begin{definition}[Generalized unextendible entanglement]
\label{def:gen-unext-ent}
The generalized unextendible entanglement of a bipartite state $\rho_{AB}$ is defined as%
\begin{equation}
\label{eq:gd ext}
\mathbf{E}^{u}(\rho_{AB})
\coloneqq \frac{1}{2}\inf_{\sigma_{AB'}\in\cF_{\rho_{AB}}}\mathbf{D}\left(\rho_{AB}\Vert\sigma_{AB'}\right),
\end{equation}
where the infimum ranges over all free states in $\cF_{\rho_{AB}}$, as defined in \eqref{eq:free-states}.
\end{definition}

\begin{remark}[Connection to joinability]
In \cite[Definition~1]{PhysRevA.88.032323}, the concept of joinability of two bipartite states was considered. That is, two states $\sigma_{AB}$ and $\tau_{AB'}$ are said to be joinable if there exists a tripartite state $\omega_{ABB'}$ such that $\operatorname{Tr}_{B'}[\omega_{ABB'}]=\sigma_{AB}$ and $\operatorname{Tr}_{B}[\omega_{ABB'}]=\tau_{AB'}$. With this in mind, the measure in \eqref{eq:gd ext} can be understood as quantifying how unjoinable a bipartite state~$\rho_{AB}$ is with other bipartite states of systems $A$ and $B'$. 
\end{remark}


\begin{remark}[Channel representation]\label{remark:purification-extension}
In~\cite{Christandl2004}, it was established that a state $\rho_{ABB'}$ is an extension of $\rho_{AB}$ if and
only if there exists a quantum channel $\cR_{C\to B'}$ such that $\rho_ {ABB'} =
\cR_{C\to B'}(\Psi_{ABC})$, where  $\Psi_{ABC}$ is a purification of $\rho_{AB}$. Using this correspondence, the generalized unextendible entanglement can be defined in a
dynamical way as
\begin{equation}
\mathbf{E}^{u}(\rho_{AB})
\coloneqq  \frac{1}{2}\inf_{\cR_{C\to B'}}\left\{\mathbf{D}(\rho_{AB}\Vert\rho_{AB'}):
\rho_{ABB'} = \left(\id_{AB}\ox\cR_{C\to B'}\right)\left(\proj{\Psi}_{ABC}\right)\right\},
\end{equation}
where the infimum ranges over every quantum channel $\cR_{C\to B'}$.
\end{remark}

In the following, we show that the generalized unextendible entanglement possesses
the desirable monotonicity property required for a valid entanglement measure~\cite{Horodecki2009}.
Notice that we have only proven the monotonicity for (selective) two-extendible operations,
which include (selective) 1-LOCC operations as special cases.
We do not expect the unextendible entanglement to be monotone under all LOCC operations.

\begin{theorem}
[Two-extendible monotonicity]
\label{thm:gen-div-two-ext-mono}
The generalized unextendible entanglement does not increase
under the action of a two-extendible channel. That is, the following
inequality holds for an arbitrary bipartite quantum state $\rho_{AB}$
and an arbitrary two-extendible channel $\cE_{AB\rightarrow A^{\prime}B^{\prime}}$:
\begin{equation}
\mathbf{E}^{u}(\rho_{AB})\geq\mathbf{E}^{u}(\cE_{AB\rightarrow A^{\prime
}B^{\prime}}(\rho_{AB})).\label{eq:-gen-unext-monotone}%
\end{equation}
\end{theorem}

\begin{proof}
Let $\rho_{AB_{1}B_{2}}$ be an arbitrary extension of $\rho_{AB}$, with $\rho_{AB}=\rho_{AB_{1}%
}$. Since the channel $\cE_{AB\rightarrow AB^{\prime}}$ is two-extendible,
there exists an extending channel $\cE_{AB_{1}B_{2}\rightarrow
AB_{1}^{\prime}B_{2}^{\prime}}$ that satisfies the conditions stated in \eqref{eq:marginal}--\eqref{eq:perm-cov}.
In particular, the marginal channels $\cE_{AB_{1}\rightarrow
AB_{1}^{\prime}}$ and $\cE_{AB_{2}\rightarrow AB_{2}^{\prime}}$ are each
equal to the original channel $\cE_{AB\rightarrow AB^{\prime}}$.
Furthermore, the state $\cE_{AB_{1}B_{2}\rightarrow AB_{1}^{\prime}%
B_{2}^{\prime}}(\rho_{AB_{1}B_{2}})$ is an extension of $\cE
_{AB_{1}\rightarrow AB_{1}^{\prime}}(\rho_{AB_{1}})$ and $\cE
_{AB_{2}\rightarrow AB_{2}^{\prime}}(\rho_{AB_{2}})$. Then we conclude that%
\begin{align}
\mathbf{D}(\rho_{AB_{1}}\Vert\rho_{AB_{2}})  & \geq\mathbf{D}(\cE
_{AB\rightarrow AB^{\prime}}(\rho_{AB_{1}})\Vert\cE_{AB\rightarrow
AB^{\prime}}(\rho_{AB_{2}}))\\
& =\mathbf{D}(\cE_{AB_{1}\rightarrow AB_{1}^{\prime}}(\rho_{AB_{1}}%
)\Vert\cE_{AB_{2}\rightarrow AB_{2}^{\prime}}(\rho_{AB_{2}}))\\
& \geq2\mathbf{E}^{u}(\cE_{AB\rightarrow A^{\prime}B^{\prime}}(\rho
_{AB})).
\end{align}
The first inequality follows from monotonicity of the generalized channel
divergence under the action of the quantum channel $\cE_{AB\rightarrow
AB^{\prime}}$. The equality follows from the observations stated above. The
final inequality follows from the definition of the generalized unextendible
entanglement, by applying the fact that $\cE_{AB_{1}B_{2}\rightarrow
AB_{1}^{\prime}B_{2}^{\prime}}(\rho_{AB_{1}B_{2}})$ is an extension of
$\cE_{AB_{1}\rightarrow AB_{1}^{\prime}}(\rho_{AB_{1}})$ and
$\cE_{AB_{2}\rightarrow AB_{2}^{\prime}}(\rho_{AB_{2}})$. Since the
inequality above holds for an arbitrary extension $\rho_{AB_{1}B_{2}}$ of
$\rho_{AB}$, we can take an infimum over all such extensions and conclude \eqref{eq:-gen-unext-monotone}.
\end{proof}

\bigskip
It turns out that if the underlying generalized divergence is faithful
and continuous, then the generalized unextendible entanglement $\mathbf{E}^{u}$ is faithful,
in the sense that $\mathbf{E}^{u}(\rho_{AB}) = 0$ if and only if $\rho_{AB}$ is two-extendible.

\begin{proposition}[Faithfulness]
\label{prop:faithfulness}
If the underlying generalized divergence is faithful, then the generalized unextendible entanglement $\mathbf{E}^{u}$ is non-negative and $\mathbf{E}^{u}(\rho_{AB}) = 0$ if $\rho_{AB}$ is two-extendible.
If the underlying generalized divergence is also continuous, then the state $\rho_{AB}$ is two-extendible if $\mathbf{E}^{u}(\rho_{AB}) = 0$.
\end{proposition}

\begin{proof}
Non-negativity of $\mathbf{E}^{u}$ follows from the assumption of faithfulness. 
Now suppose that $\rho_{AB}$ is two-extendible.
Then there exists an extension $\rho_{ABB'}$ such that $\rho_{AB}= \rho_{AB'}$. 
Then it follows that $\mathbf{E}^{u}(\rho_{AB}) = 0$ by definition and 
from the assumption that the underlying divergence is faithful.

Now suppose that $\mathbf{E}^{u}(\rho_{AB}) = 0$. By the assumption of continuity, this means that 
there exists an extension $\rho_{ABB'}$ of $\rho_{AB}$ 
such that $\rho_{AB}= \rho_{AB'}$. Then, by definition, $\rho_{AB}$ is two-extendible.
\end{proof}

\subsection{$\alpha$-unextendible entanglement}
\label{sec:alpha unext}

The \textit{$\alpha$-unextendible entanglement} of a bipartite state $\rho_{AB}$\ is defined by taking the generalized divergence in \eqref{eq:gd ext} to be the Petz--R\'enyi relative entropy \cite{petz1986quasi}, the sandwiched R\'enyi relative entropy $\widetilde{D}_{\alpha}%
(\omega\Vert\tau)$ \cite{muller2013quantum,wilde2014strong},
or the geometric R\'enyi relative entropy \cite{matsumoto2015}.
It is meaningful to consider the families of $\alpha$-unextendible entanglement measures
with different quantum divergences for several reasons. 
First of all, we are able prove their general properties once and for all, without analyzing them case by case.
Second, some special cases of these measures lead to novel applications, as 
we will show in Section~\ref{sec:applications}. 
It is highly possible that more of them find operational interpretations in other tasks, 
using the insights learned from the proved properties.

\begin{definition}[$\alpha$-unextendible entanglement]
The $\alpha$-Petz unextendible entanglement is defined for $\alpha \in [0,\infty)$ as
\begin{equation}
E_{\alpha}^{u}(\rho_{AB})
\coloneqq \frac{1}{2}\inf_{\sigma_{AB'}\in\cF_{\rho_{AB}}}D_{\alpha}\left(\rho_{AB}\Vert\sigma_{AB'}\right),
\end{equation}
the $\alpha$-sandwiched unextendible entanglement is defined for $\alpha
\in(0,\infty]$ as
\begin{equation}
\widetilde{E}_{\alpha}^{u}(\rho_{AB})
\coloneqq \frac{1}{2}\inf_{\sigma_{AB'}\in\cF_{\rho_{AB}}}\widetilde{D}_{\alpha}\left(\rho_{AB}\Vert\sigma_{AB'}\right),
\end{equation}
and the $\alpha$-geometric unextendible entanglement is defined for $\alpha
\in(0,\infty)$ as
\begin{equation}
\widehat{E}_{\alpha}^{u}(\rho_{AB})
\coloneqq \frac{1}{2}\inf_{\sigma_{AB'}\in\cF_{\rho_{AB}}}\widehat{D}_{\alpha}\left(\rho_{AB}\Vert\sigma_{AB'}\right).
\end{equation}
\end{definition}

In the following, we show that the $\alpha$-unextendible entanglement possesses
the desirable monotonicity property required for an entanglement measure~\cite{Horodecki2009}. More specifically, it is an immediate consequence of Theorem~\ref{thm:gen-div-two-ext-mono} and data processing that $\alpha$-unextendible entanglement is monotone under the action of a two-extendible channel for certain values of $\alpha$.

\begin{corollary}[Two-extendible monotonicity]\label{coro:two-extendible-monotonicity}
For $\alpha \in [0,2]$, the $\alpha$-Petz unextendible entanglement does not increase under the action of a two-extendible channel. For $\alpha \in [1/2,\infty]$, the $\alpha$-sandwiched unextendible entanglement does not increase under the action of a two-extendible channel.
For $\alpha \in (0,2]$, the $\alpha$-geometric unextendible entanglement does not increase under the action of a two-extendible channel.
\end{corollary}

Notice that we have only asserted the monotonicity for two-extendible operations,
including 1-LOCC operations as special cases.
As already mentioned, we do not expect the unextendible entanglement to be monotone under all LOCC operations.

For certain values of $\alpha \geq 1$, the above monotonicity statement can be strengthened,
in the sense that the $\alpha$-unextendible entanglement is monotone under \textit{selective} two-extendible 
operations. This property is stronger than what we previously proved in Theorem~\ref{thm:gen-div-two-ext-mono}. 
As before, this implies that $\alpha$-unextendible entanglement is monotone under selective 1-LOCC operations 
because these are a particular case of selective two-extendible operations.
The proof of Theorem~\ref{thm:strong-2-ext-mono} can be found in Appendix~\ref{appx:strong-2-ext-mono}.

\begin{theorem}[Selective two-extendible monotonicity]
\label{thm:strong-2-ext-mono}
Let $\rho_{AB}$ be a bipartite state, and let $\left\{  \mathcal{E}%
_{AB\rightarrow A^{\prime}B^{\prime}}^{y}\right\}  _{y}$ be a selective
two-extendible operation. Then the following inequality holds for the $\alpha$-Petz unextendible entanglement
for all
$\alpha\in\lbrack1,2]$:%
\begin{equation}
E_{\alpha}^{u}(\rho_{AB})\geq\sum_{y:p(y)>0}p(y)E_{\alpha}^{u}(\omega
_{A^{\prime}B^{\prime}}^{y}), \label{eq:full-monotone-ineq}%
\end{equation}
where%
\begin{align}
p(y)  &  \coloneqq \operatorname{Tr}[\mathcal{E}_{AB\rightarrow A^{\prime}B^{\prime}%
}^{y}(\rho_{AB})],\label{eq:p(y)-thm}\\
\omega_{A^{\prime}B^{\prime}}^{y}  &  \coloneqq \frac{1}{p(y)}\mathcal{E}%
_{AB\rightarrow A^{\prime}B^{\prime}}^{y}(\rho_{AB}).\label{eq:omega-thm}
\end{align}
The following inequality holds for the
$\alpha$-sandwiched unextendible entanglement
for all
$\alpha\in\lbrack1,\infty]$:%
\begin{equation}
\widetilde{E}_{\alpha}^{u}(\rho_{AB})\geq\sum_{y:p(y)>0}p(y)\widetilde{E}_{\alpha}^{u}(\omega
_{A^{\prime}B^{\prime}}^{y}). \label{eq:full-monotone-ineq-sandwiched}%
\end{equation}
The following inequality holds for the
$\alpha$-geometric unextendible entanglement
for all
$\alpha\in\lbrack1,2]$:%
\begin{equation}
\widehat{E}_{\alpha}^{u}(\rho_{AB})\geq\sum_{y:p(y)>0}p(y)\widehat{E}_{\alpha}^{u}(\omega
_{A^{\prime}B^{\prime}}^{y}). \label{eq:full-monotone-ineq-geometric}%
\end{equation}
\end{theorem}

Since the Petz--, sandwiched, and geometric R\'enyi relative entropies
(and their limits in $\alpha\to1$) satisfy all of the requirements in Lemma~\ref{lem:faithfulness-renyi-rel-ents}, 
we can invoke Proposition~\ref{prop:faithfulness} and conclude the faithfulness property
for all $\alpha$-unextendible entanglement measures.

\begin{proposition}[Faithfulness]
For $\alpha \in (0,\infty)$, the $\a$-Petz unextendible entanglement of a bipartite state $\rho_{AB}$ satisfies $E_{\alpha}^u(\rho_{AB}) \ge 0$, and
 $E_{\alpha}^u(\rho_{AB}) = 0$ if and only if $\rho_{AB}$ is two-extendible.
For $\alpha \in (0,\infty]$, the $\a$-sandwiched unextendible entanglement of a bipartite state $\rho_{AB}$ satisfies $\widetilde{E}_{\alpha}^u(\rho_{AB}) \ge 0$, and
 $\widetilde{E}_{\alpha}^u(\rho_{AB}) = 0$ if and only if $\rho_{AB}$ is two-extendible.
For $\alpha \in (0,\infty)$, the $\a$-geometric unextendible entanglement of a bipartite state $\rho_{AB}$ satisfies $\widehat{E}_{\alpha}^u(\rho_{AB}) \ge 0$, and
 $\widehat{E}_{\alpha}^u(\rho_{AB}) = 0$ if and only if $\rho_{AB}$ is two-extendible.
\end{proposition}

Before stating the next proposition, recall that the Petz--R\'enyi mutual information of a bipartite state $\rho_{AB}$ is defined as \cite[Eq.~(6.3) and Corollary~8]{GW15}
\begin{equation}
I_{\alpha}(A;B)_{\rho} \coloneqq  \min_{\sigma_{B}}
D_{\alpha}(\rho_{AB}\| \rho_A \ox \sigma_{B})
= \frac{\alpha}{\alpha-1}\log \tr\!\left[ \left(\tr_A[\rho_A^{1-\alpha}\rho_{AB}^{\alpha}]\right)^{1/\alpha}\right],
\label{eq:Petz-Renyi-MI}
\end{equation}
and the sandwiched R\'enyi mutual information of a bipartite state $\rho_{AB}$ is defined as \cite{Bei13,GW15}
\begin{equation}
\widetilde{I}_{\alpha}(A;B)_{\rho} \coloneqq  \min_{\sigma_{B}}
\widetilde{D}_{\alpha}(\rho_{AB}\| \rho_A \ox \sigma_{B}),
\label{eq:sandwiched-Renyi-MI}
\end{equation}
where the minimization in both cases is with respect to every quantum state $\sigma_{B}$.
We also define the geometric R\'enyi mutual information of $\rho_{AB}$ in a similar way:
\begin{equation}
\widehat{I}_{\alpha}(A;B)_{\rho} \coloneqq  \min_{\sigma_{B}}
\widehat{D}_{\alpha}(\rho_{AB}\| \rho_A \ox \sigma_{B}).
\label{eq:geometric-Renyi-MI}
\end{equation}

Also, recall that the R\'enyi entropy of a quantum state $\rho_A$ is
defined for $\alpha\in(0,1)\cup(1,\infty)$ as
\begin{equation}
H_{\alpha}(A)_\rho \coloneqq  \frac{1}{1-\alpha} \log  \tr[\rho_A^{\alpha}],
\end{equation}
and it is defined for $\alpha \in \{0,1,\infty\}$ in the limit, so that
\begin{align}
H_0(A)_\rho & = \log  \operatorname{rank}(\rho_A), \\
H(A)_\rho 	&\coloneqq  H_1(A)_\rho = -\tr[\rho_A \log  \rho_A],\\
H_{\min}(A)_{\rho} & \coloneqq  H_{\infty}(A)_{\rho} = -\log  \left\Vert \rho_A\right\Vert_{\infty},
\end{align}
where $\operatorname{rank}(\rho_A)$ is the rank of $\rho_A$.

The following lemmas relate the $\alpha$-mutual information to R\'enyi entropy when the bipartite state is pure. The first was established in the proof of \cite[Proposition~13]{SBW14}.

\begin{lemma}[\cite{SBW14}]
\label{lem:Petz-Renyi-MI-pure}
Let $\alpha \in (0,\infty)$. The $\alpha$-Petz R\'enyi mutual information of a pure bipartite state
$\psi_{AB}$ reduces to twice the $\gamma$-R\'enyi
entropy of its marginal:%
\begin{equation}
I_{\alpha}(A;B)_{\psi}=2H_{\gamma}(A)_{\psi},
\end{equation}
where $\gamma\coloneqq \left[  2-\alpha\right]/\alpha  $.
\end{lemma}

\begin{lemma}
\label{lem:sandwiched-Renyi-MI-pure}
Let $\alpha \in (0,\infty]$. The $\alpha$-sandwiched R\'enyi mutual information of a pure bipartite state
$\psi_{AB}$ reduces to twice the $\beta$-R\'enyi
entropy of its marginal:%
\begin{equation}
\widetilde{I}_{\alpha}(A;B)_{\psi}=2H_{\beta}(A)_{\psi},
\end{equation}
where $\beta\coloneqq \left[  2\alpha-1\right]  ^{-1}$.
\end{lemma}

\begin{proof}
See Appendix~\ref{app:proof-lemma-sw-Rny-MI-pure}.
\end{proof}

\begin{lemma}
\label{lem:geometric-Renyi-MI-pure}
Let $\alpha \in (0,\infty)$. The $\alpha$-geometric R\'enyi mutual information of a pure bipartite state
$\psi_{AB}$ reduces to twice the zero-R\'enyi
entropy of its marginal:%
\begin{equation}
\widehat{I}_{\alpha}(A;B)_{\psi}=2H_{0}(A)_{\psi}.
\end{equation}
\end{lemma}

\begin{proof}
See Appendix~\ref{app:proof-lemma-geometric-Rny-MI-pure}.
\end{proof}

\bigskip
The following proposition uses the  lemmas above to conclude that the $\alpha$-unextendible
entanglement reduces to R\'enyi entropy of entanglement for pure states.

\begin{proposition}[Reduction for pure states]
\label{prop:reduction-for-pure-states}
Let $\psi_{AB}$ be a pure bipartite state. Then the $\alpha$-Petz unextendible entanglement reduces to the $\gamma$-R\'enyi entropy of entanglement:
\begin{equation}
E_{\alpha}(\psi_{AB}) = H_{\gamma}(\psi_A),
\label{eq:e-alpha-unext-pure-states}
\end{equation}
where $\gamma = [2  - \alpha]/\alpha$. The $\alpha$-sandwiched unextendible entanglement reduces to the $\beta$-R\'enyi entropy of entanglement:
\begin{equation}
\widetilde{E}_{\alpha}(\psi_{AB}) = H_{\beta}(\psi_A),
\end{equation}
where $\beta = [2 \alpha - 1]^{-1}$.
The $\alpha$-geometric unextendible entanglement reduces to the zero-R\'enyi entropy of entanglement:
\begin{equation}
\widehat{E}_{\alpha}(\psi_{AB}) = H_{0}(\psi_A).
\end{equation}
\end{proposition}

\begin{proof}
For a pure state $\psi_{AB}$, an arbitrary extension of it has  the form
$\sigma_{AB_1B_2}\coloneqq \psi_{AB_1}\ox\sigma_{B_2}$, where $\sigma_{B_2}$ is a state of system $B_2$.
As such, it follows that
\begin{align}
E_{\alpha}(\psi_{AB})  & = \min_{\sigma_{B_2}}\frac{1}{2}
D_{\alpha}(\psi_{AB}\| \psi_A \ox \sigma_{B_2}) \\
& = \frac{1}{2} I_{\alpha}(A;B)_{\psi} \\
& = H_{\gamma}(\psi_A).
\end{align}
The first equality follows from applying the definition of $E_{\alpha}(\psi_{AB})$,
the second equality follows from the definition in \eqref{eq:Petz-Renyi-MI},
and the final equality follows from Lemma~\ref{lem:Petz-Renyi-MI-pure}.

The conclusions about $\widetilde{E}_{\alpha}(\psi_{AB})$ and $\widehat{E}_{\alpha}(\psi_{AB})$ follow the same line of reasoning but using Lemma~\ref{lem:sandwiched-Renyi-MI-pure} and Lemma~\ref{lem:geometric-Renyi-MI-pure} instead, respectively.
\end{proof}

\bigskip

Following Proposition~\ref{prop:reduction-for-pure-states} and the fact that the R\'enyi 
entropy of the maximally mixed state (reduced state of $\Phi^d_{AB}$) 
is equal to $\log d$ for all values of $\beta$, we conclude the normalization property and state it formally as follows.

\begin{proposition}[Normalization]
\label{prop:normalization}
Let $\Phi^d_{AB}$ be a maximally entangled state of Schmidt rank $d$. Then $E_{\alpha}(\Phi^d_{AB}) = \log d$ for $\alpha \in (0,\infty)$, $\widetilde{E}_{\alpha}(\Phi^d_{AB}) = \log d$ for $\alpha \in (0,\infty]$, and $\widehat{E}_{\alpha}(\Phi^d_{AB}) = \log d$ for $\alpha \in (0,\infty)$.
\end{proposition}

\bigskip
Let $\{\rho^i_{AB}\}_i\subset\cS(AB)$ be a set of bipartite quantum states,
and let $\{p_i\}_i$ be a probability distribution such that $p_i\geq0$ and $\sum_ip_i=1$.
A generalized unextendible entanglement $\mathbf{E}^u$ is \textit{convex} if
\begin{align}
    	\mathbf{E}^u\!\left(\sum_ip_i\rho^i_{AB}\right)
\leq  \sum_ip_i\mathbf{E}^u(\rho^i_{AB}),
\end{align}
and it is \textit{quasi-convex} if
\begin{align}
    	\mathbf{E}^u\!\left(\sum_ip_i\rho^i_{AB}\right)
\leq  \max_i\mathbf{E}^u(\rho^i_{AB}).
\end{align}
If $\mathbf{E}^u$ is convex, then it is also quasi-convex; however, the converse is not necessarily true.

Intuitively, the convexity (quasi-convexity) property implies
that the amount of quantum entanglement cannot be increased by merely mixing the quantum states.
Convexity is a desirable property of a valid entanglement measure from both the mathematical and physical perspectives.
The $\a$-unextendible entanglement satisfies this property for certain parameter ranges:

\begin{proposition}[Convexity and quasi-convexity]
The $\a$-Petz unextendible entanglement $E_{\alpha}^u$ 
is convex when $\alpha\in(0, 1]$ and quasi-convex when $\alpha\in(0,2]$.
The $\a$-sandwiched unextendible entanglement $\widetilde{E}_{\alpha}^u$ 
is convex when $\alpha\in[1/2, 1]$ and quasi-convex when $\alpha\geq1/2$.
The $\a$-geometric unextendible entanglement $\widehat{E}_{\alpha}^u$ 
is convex when $\alpha\in(0, 1]$ and quasi-convex when $\alpha\in(0,2]$.
\end{proposition}

\begin{proof}
This follows from the facts that $D_{\alpha}$ is jointly quasi-convex for $\alpha \in(1,2]$ and
jointly convex for $\alpha\in(0, 1]$ \cite{petz1986quasi};
$\widetilde{D}_{\alpha}$ is jointly quasi-convex for $\alpha>1$ and jointly convex for $\alpha\in[1/2, 1]$ \cite{FL13};
and $\widehat{D}_{\alpha}$ is jointly quasi-convex for $\alpha \in(1,2]$ and
jointly convex for $\alpha\in(0, 1]$ \cite{HM17}.
\end{proof}

\bigskip
Aside from the above properties, the $\alpha$-unextendible entanglement is also
subadditive (proof in Appendix~\ref{appx:subadditivity}).

\begin{proposition}[Subadditivity]\label{prop:subadditivity}
Let $\rho_{A_1B_1}\in\cS(A_1B_1)$ and $\rho_{A_2B_2}\in\cS(A_2B_2)$ be two bipartite states. Then the following subadditivity inequality holds for $\a \in [0,\infty)$
\begin{equation}
      E_{\a}^u(\rho_{A_1B_1}\ox\rho_{A_2B_2})
\leq  E_{\a}^u(\rho_{A_1B_1}) + E_{\a}^u(\rho_{A_2B_2}),
\label{eq:petz-subadd}
\end{equation}
for  $\a \in (0,\infty]$
\begin{equation}
      \widetilde{E}_{\a}^u(\rho_{A_1B_1}\ox\rho_{A_2B_2})
\leq  \widetilde{E}_{\a}^u(\rho_{A_1B_1}) + \widetilde{E}_{\a}^u(\rho_{A_2B_2}),
\label{eq:sandwiched-subadd}
\end{equation}
and for  $\a \in (0,\infty)$
\begin{equation}
      \widehat{E}_{\a}^u(\rho_{A_1B_1}\ox\rho_{A_2B_2})
\leq  \widehat{E}_{\a}^u(\rho_{A_1B_1}) + \widehat{E}_{\a}^u(\rho_{A_2B_2}),
\label{eq:geometric-subadd}
\end{equation}
where the entanglement is evaluated across the $A{:}B$ bipartition.
\end{proposition}

\subsection{Relative-entropy-induced unextendible entanglement}

As mentioned above, the quantum relative entropy is a particular instance of the Petz- and
sandwiched R\'{e}nyi relative entropy, recovered by taking the limit $\alpha \to 1$. The $\alpha$-unextendible entanglement
in terms of these measures has already been defined and investigated above.
However, it is notable enough that we define
the quantum relative entropy induced unextendible entanglement explicitly here.

\begin{definition}[Relative-entropy-induced unextendible entanglement]\label{def:rel-unextendible}
For a bipartite state~$\rho_{AB}$, the quantum relative entropy induced unextendible entanglement is
defined as
\begin{align}\label{eq:rel-unextendible}
E^u(\rho_{AB})
\coloneqq \frac{1}{2}\inf_{\sigma_{AB'}\in\cF_{\rho_{AB}}}D\left(\rho_{AB}\Vert\sigma_{AB'}\right),
\end{align}
where $D$ is the quantum relative entropy defined in~\eqref{eq:rel-ent}.
\end{definition}

As proved above, $E^u$ obeys the following properties: selective two-extendible monotonicity,
faithfulness, reduction to the entropy of entanglement for pure bipartite states, normalization, convexity, and subadditivity. It is also efficiently computable by means of semi-definite programming, using the approach of \cite{fawzi2019semidefinite,Fawzi_2018}.

The third property mentioned above follows directly from
Proposition~\ref{prop:reduction-for-pure-states} for $\alpha=1$,
and it asserts that for pure bipartite states, the relative-entropy-induced unextendible entanglement
evaluates to the von Neumann entropy of the reduced state. Although the proof has already been given,
we can see it with a straightforward
proof consisting of a few steps. Let $\psi_{AB}\equiv\ket{\psi}\!\bra{\psi}_{AB}$ be a pure state. An arbitrary extension of $\psi_{AB}$ is of the form
$\sigma_{ABB'}\coloneqq \psi_{AB}\ox\sigma_{B'}$ for some state $\sigma_{B'}$. As so,
$\tr_{B}[\sigma_{ABB'}] = \psi_A\ox\sigma_{B'}$, where $\psi_A = \tr_B[\psi_{AB}]$. Then
\begin{align}
    E^u(\psi_{AB})
= \inf_{\sigma_{B'}}\frac{1}{2} D(\psi_{AB}\| \psi_A\ox\sigma_{B'})
=  \frac{1}{2}I(A;B)_\psi = H(A)_\psi,
\end{align}
where the second equality follows from the definition of quantum mutual information.


\section{Efficiently computable entanglement measures}

\label{sec:efficiently-comp-unext}

\subsection{Max-unextendible entanglement}

Another interesting instance of the sandwiched R\'{e}nyi relative entropy is the max-relative entropy, 
as recalled in \eqref{eq:max-rel-ent-limit}.
The max-relative entropy was originally defined and studied
in~\cite{Datta2009}. Here we adopt the max-relative entropy to define the max-unextendible entanglement. It turns out that
this measure is additive and can be calculated efficiently by utilizing a semidefinite program.

\begin{definition}[Max-unextendible entanglement]\label{def:max-unextendible}
For a given bipartite state $\rho_{AB}$, the max-unextendible entanglement is defined as
\begin{align}\label{eq:max-unextendible}
E_{\max}^u(\rho_{AB})
\coloneqq \frac{1}{2}\inf_{\sigma_{AB'}\in\cF_{\rho_{AB}}}D_{\max}\left(\rho_{AB}\Vert\sigma_{AB'}\right),
\end{align}
where $D_{\max}$ is the max-relative entropy defined in~\eqref{eq:max-rel-ent-limit}.
\end{definition}

Note that the infimum in~\eqref{eq:max-unextendible} can be replaced with a minimum.

From the definition of max-unextendible entanglement, it follows that it can be computed efficiently by means of a
semidefinite program (SDP). To be more specific, the following two optimization programs satisfy strong duality and
both evaluate to $2^{-2E_{\max}^u(\rho_{AB})}$.

\begin{minipage}[t]{0.45\linewidth}
\begin{equation}\label{eq:max-primal}
\begin{split}
& \text{\underline{\bf Primal Program}} \\
\text{maximize}   &\quad \lambda \\
\text{subject to} &\quad \lambda\rho_{AB}\leq\sigma_{AB'} \\
            &\quad \rho_{AB} = \tr_{B'}[\sigma_{ABB'}] \\
            &\quad \sigma_{ABB'} \geq 0
\end{split}\end{equation}
\end{minipage}
\begin{minipage}[t]{0.45\linewidth}
\begin{equation}\label{eq:max-dual}
\begin{split}
& \text{\underline{\bf Dual Program}} \\
\text{minimize}   &\quad \tr\left[\rho_{AB}Y_{AB}\right] \\
\text{subject to} &\quad \tr\left[\rho_{AB}X_{AB}\right] \geq 1 \\
                &\quad X_{AB'}\ox \1_B \leq Y_{AB}\ox\1_{B'} \\
                &\quad X_{AB}, Y_{AB} \geq 0
\end{split}
\end{equation}
\end{minipage}

\bigskip
The primal SDP follows by considering that
\begin{equation}
D_{\max}(\omega\Vert \tau) = \log \inf\{\lambda : \omega \leq \lambda \tau\} = -\log \sup\{\mu : \mu \omega \leq  \tau\}.
\end{equation}
The dual SDP can be obtained by standard methods (e.g., the Lagrange multiplier method). For completeness, we provide a proof in Appendix~\ref{app:dual-SDP-derivations}.

By using the primal and dual expressions of $2^{-2E_{\max}^u(\rho_{AB})}$ and strong duality, it follows that
$E_{\max}^u(\rho_{AB})$ is additive (proof in Appendix~\ref{appx:max-dditivity}), 
which is an appealing feature that finds use in Section~\ref{sec:applications}.

\begin{proposition}[Additivity]
\label{prop:max-dditivity}
Let $\rho_{A_1B_1}\in\cS(A_1B_1)$ and $\rho_{A_2B_2}\in\cS(A_2B_2)$ be two bipartite states.
It holds that
\begin{align}\label{eq:max-dditivity}
  E_{\max}^u\left(\rho_{A_1B_1}\ox\rho_{A_2B_2}\right)
= E_{\max}^u\left(\rho_{A_1B_1}\right) + E_{\max}^u\left(\rho_{A_2B_2}\right),
\end{align}
where the entanglement is evaluated across the $A{:}B$ bipartition.
\end{proposition}


\subsection{Min-unextendible entanglement}

In this section, we consider the limit of the Petz--R\'enyi relative entropy as $\alpha \to 0$, which is known as the
min-relative entropy~\cite[Definition 2]{Datta2009}. Let us first recall the definition. Let $\rho\in\cS(A)$ and
$\sigma\in\cP(A)$. Let $\Pi^\rho$ denote the
projection onto the support of $\rho$. Then the min-relative entropy of $\rho$ and $\sigma$ is defined as
\begin{align}
    D_{\min}\left(\rho\|\sigma\right) \coloneqq  -\log \tr\left[\Pi^\rho\sigma\right].
\end{align}
With $D_{\min}$, we define the min-unextendible entanglement as follows.
\begin{definition}[Min-unextendible entanglement]\label{def:min-unextendible}
For a given bipartite state $\rho_{AB}$, the min-unextendible entanglement is defined as
\begin{align}\label{eq:min-unextendible}
E_{\min}^u(\rho_{AB})
\coloneqq \frac{1}{2}\inf_{\sigma_{AB'}\in\cF_{\rho_{AB}}}D_{\min}\left(\rho_{AB}\Vert\sigma_{AB'}\right).
\end{align}
\end{definition}

Note that the infimum in~\eqref{eq:min-unextendible} can be replaced with a minimum.

Much like the max-unextendible entanglement, the min-unextendible entanglement can also be calculated as 
the solution to a semidefinite program. 
The following two optimization programs satisfy strong duality, and both evaluate to
$2^{-2E_{\min}^u(\rho_{AB})}$. We derive the dual SDP in Appendix~\ref{app:dual-SDP-derivations}.

\begin{minipage}[t]{0.45\linewidth}
\begin{equation}\label{eq:min-primal}
\begin{split}
& \text{\underline{\bf Primal Program}} \\
\text{maximize}   &\quad \tr[\Pi^\rho_{AB}\sigma_{AB'}] \\
\text{subject to} &\quad \tr_{B'} [\sigma_{ABB'}] = \rho_{AB}\\
						&\quad \sigma_{ABB'} \geq 0
\end{split}\end{equation}
\end{minipage}
\begin{minipage}[t]{0.45\linewidth}
\begin{equation}\label{eq:min-dual}
\begin{split}
& \text{\underline{\bf Dual Program}} \\
\text{minimize} 	&\quad \tr\left[\rho_{AB}X_{AB}\right] \\
\text{subject to} &\quad X_{AB}\ox \1_{B'} \geq\Pi^\rho_{AB'}\ox\1_{B} \\
						& \quad X_{AB} \geq 0
\end{split}
\end{equation}
\end{minipage}

\bigskip
Following similarly to the proof of Proposition~\ref{prop:max-dditivity},
we can show that the min-unextendible entanglement is additive.

\begin{proposition}[Additivity]
\label{prop:add-min-unext-ent}
Let $\rho_{A_1B_1}\in\cS(A_1B_1)$ and $\rho_{A_2B_2}\in\cS(A_2B_2)$ be two bipartite states.
It holds that
\begin{align}
  E_{\min}^u(\rho_{A_1B_1}\ox\rho_{A_2B_2})
= E_{\min}^u(\rho_{A_1B_1}) + E_{\min}^u(\rho_{A_2B_2}),
\end{align}
where the entanglement is evaluated across the $A{:}B$ bipartition.
\end{proposition}

Furthermore, the min-unextendible entanglement of a pure bipartite state 
can be computed explicitly (proof in Appendix~\ref{appx:e-u-min-pure}).

\begin{proposition}
\label{prop:e-u-min-pure}
Let $\ket{\psi}_{AB}\coloneqq \sum_{i=1}^k\sqrt{\alpha_i}\ket{\psi_i}_A \ket{\psi_i}_B$ be a pure state in
$\cH_{AB}$, with $\alpha_1\geq\cdots \geq \alpha_k>0$. Then the following equality holds
\begin{align}
    E_{\min}^u\left(\psi_{AB}\right) = - \log\alpha_1.
\end{align}
\end{proposition}

\bigskip
Interestingly, $E_{\min}^u(\psi_{AB})$ has an operational interpretation in terms of
deterministic entanglement transformation~\cite{Duan2005b}, which we briefly introduce as follows. Let
$\ket{\psi_1}$ and $\ket{\psi_2}$ be two pure bipartite states for systems $AB$. Let $m\in\mathbb{N}$ be an integer. We define
$f(m)$ to be the maximum integer $n$ such that $\psi_1^{\ox m}$ can be transformed into $\psi_2^{\ox n}$ by LOCC
deterministically. The \textit{deterministic entanglement transformation rate} from $\psi_1$ to $\psi_2$, written
$D(\psi_1\to\psi_2)$, is defined as
\begin{align}
    D(\psi_1\to\psi_2) \coloneqq  \sup_{m\geq1} \frac{f(m)}{m}.
\end{align}
Intuitively, for sufficiently large $m$, one can transform $m$ copies of $\psi_1$ exactly into $mD(\psi_1\to\psi_2)$
copies of $\psi_2$ by LOCC. We have the following proposition, which is a consequence of Proposition~\ref{prop:e-u-min-pure} and the developments in \cite{Duan2005b}:
\begin{proposition}
Let $\ket{\psi}_{AB}$ be a pure state in $AB$ and $\Phi^2$ be the ebit. Then the following equality holds
\begin{align}
    D\left(\psi\to\Phi^2\right) = E_{\min}^u\left(\psi_{AB}\right).
\end{align}
\end{proposition}


\subsection{Unextendible fidelity}

Let $\rho,\sigma\in\cS(A)$ be two quantum states. The (root) fidelity between $\rho$ and $\sigma$ is defined as \cite{U76}
\begin{align}
  F(\rho,\sigma) \coloneqq  \left\Vert\sqrt{\rho}\sqrt{\sigma}\right\Vert_1
  								= \tr\!\left[\sqrt{\sqrt{\sigma}\rho\sqrt{\sigma}}\right].\label{eq:fidelity}
\end{align}
Here we define the unextendible fidelity of a state $\rho_{AB}$:

\begin{definition}[Unextendible fidelity]
For a given bipartite state $\rho_{AB}$, the unextendible fidelity is defined as
\begin{align}\label{eq:unextendible fidelity}
F^u(\rho_{AB})\coloneqq \sup_{\sigma_{AB'}\in\cF_{\rho_{AB}}}F(\rho_{AB}, \sigma_{AB'}).
\end{align}
\end{definition}

Note that the supremum in~\eqref{eq:unextendible fidelity} can be replaced with a maximum.

Suppose that $\Psi_{ABC}$ is a purification of $\rho_{AB}$. By applying Remark~\ref{remark:purification-extension},
we see that the
unextendible fidelity can be alternatively understood as a measure of how well one can recover the state $\rho_{AB}$ if
system $B$ is lost and a recovery channel is performed on the purification system $C$ alone, due to the following
equivalent formulation:
\begin{align}
  F^u(\rho_{AB}) = \max_{\cR_{C\to B'}}
  \left\{F(\rho_{AB}, \rho_{AB'}): \rho_{ABB'} = \left(\id_{AB}\ox\cR_{C\to B'}\right)\left(\proj{\Psi}_{ABC}\right)
  \right\},
\end{align}
where the maximum ranges over every quantum channel $\cN_{C\to B'}$. The unextendible fidelity is thus similar in spirit to the fidelity of recovery from \cite{SW14}, but one finds that it is a different measure when analyzing it in more detail.

By examining \eqref{eq:sandwiched-Renyi} and \eqref{eq:fidelity},
one immediately finds that $\widetilde{D}_{1/2}\left(\rho\Vert\sigma\right) = - \log [F(\rho,\sigma)]^2$. Thus, we
establish the following equivalence between unextendible fidelity and $1/2$-sandwiched unextendible entanglement:
\begin{align}
    F^u(\rho_{AB})
& \coloneqq  \max_{\sigma_{AB'}\in\cF_{\rho_{AB}}}F(\rho_{AB}, \sigma_{AB'}) \label{eq:relation-unext-fid-sandwiched-1} \\
& =  \max_{\sigma_{AB'}\in\cF_{\rho_{AB}}}
       2^{-\frac{1}{2}\widetilde{D}_{1/2}\left(\rho_{AB}\Vert\sigma_{AB'}\right)} \\
& =   2^{-\widetilde{E}^u_{1/2}(\rho_{AB})}.
\label{eq:relation-unext-fid-sandwiched-last}
\end{align}

Since the fidelity function is SDP computable~\cite{watrous2013}, it follows that the unextendible fidelity can be
computed efficiently by means of a semidefinite program. To be more specific, the following two optimization programs
satisfy strong duality and both evaluate to $F^u(\rho_{AB})$. For completeness,  we show in detail
how to derive the dual program in Appendix~\ref{app:dual-sdp-unextendible-fidelity}.

\begin{minipage}[t]{0.46\linewidth}
\begin{equation}\label{eq:fid-primal}
\begin{split}
& \text{\underline{\bf Primal Program}} \\
\text{maximize}   &\quad \frac{1}{2}\tr[X_{AB}] + \frac{1}{2}\tr[X_{AB}^\dagger]  \\
\text{subject to} &\quad \begin{pmatrix} \rho_{AB} & X_{AB} \\ X_{AB}^\dagger & \tr_{B}[\sigma_{ABB'}] \end{pmatrix} \geq 0 \\
            &\quad \rho_{AB} = \tr_{B'}[\sigma_{ABB'}] \\
            &\quad \sigma_{ABB'} \geq 0 \\
            &\quad X_{AB}\in\cL(AB)
\end{split}\end{equation}
\end{minipage}
\begin{minipage}[t]{0.5\linewidth}
\begin{equation}\label{eq:fid-dual}
\begin{split}
& \text{\underline{\bf Dual Program}} \\
\text{minimize}   &\quad \frac{1}{2}\tr[W_{AB}\rho_{AB}] + \frac{1}{2}\tr[Z_{AB}\rho_{AB}] \\
\text{subject to} &\quad  \begin{pmatrix} W_{AB} & - \1_{AB} \\ -\1_{AB} & Y_{AB'}\end{pmatrix} \geq 0 \\
            &\quad Z_{AB}\ox\1_{B'} \geq Y_{AB'}\ox\1_{B} \\
            &\quad W_{AB}, Y_{AB'}, Z_{AB} \geq 0
            \end{split}\end{equation}
\end{minipage}

\bigskip

We also establish the following equivalent dual representation of $F^u(\rho_{AB})$ in
Appendix~\ref{app:dual-sdp-unextendible-fidelity}:
\begin{equation}
\begin{split}
\text{infimum}   &\quad \sqrt{\tr[Y_{AB'}^{-1}\rho_{AB}]\tr[Z_{AB}\rho_{AB}]}  \\
\text{subject to} &\quad Z_{AB}\ox\1_{B'} \geq Y_{AB'}\ox\1_{B} \\
            &\quad Z_{AB}, Y_{AB'} > 0.
\end{split}
\label{eq:dual-rep-unext-fid}
\end{equation}

As a direct consequence of this equivalent dual representation,
we find that the extendible fidelity is multiplicative (proof in Appendix~\ref{appx:fid-multiplicativity}).

\begin{proposition}[Multiplicativity]
\label{prop:fid-multiplicativity}
Let $\rho_{A_1B_1}\in\cS(A_1B_1)$ and $\rho_{A_2B_2}\in\cS(A_2B_2)$ be two bipartite states.
The following equality holds
\begin{align}\label{eq:fid-multiplicativity}
  F^u(\rho_{A_1B_1}\ox\rho_{A_2B_2}) = F^u(\rho_{A_1B_1})F^u(\rho_{A_2B_2}),
\end{align}
where the entanglement is evaluated across the $A{:}B$ bipartition.
\end{proposition}

\bigskip
Recall \eqref{eq:relation-unext-fid-sandwiched-1}--\eqref{eq:relation-unext-fid-sandwiched-last}, which relates the unextendible fidelity and the $1/2$-sandwiched unextendible entanglement.
Proposition~\ref{prop:fid-multiplicativity} demonstrates that the logarithmic unextendible fidelity is an
additive unextendible entanglement measure, different from both the min- and max-unextendible entanglement.
\begin{corollary}
Let $\rho_{A_1B_1}\in\cS(A_1B_1)$ and $\rho_{A_2B_2}\in\cS(A_2B_2)$ be two bipartite states.
The following additivity relation holds
\begin{align}
  \widetilde{E}^u_{1/2}(\rho_{A_1B_1}\ox\rho_{A_2B_2})
= \widetilde{E}^u_{1/2}(\rho_{A_1B_1}) + \widetilde{E}^u_{1/2}(\rho_{A_2B_2}),
\end{align}
where the entanglement is evaluated across the $A{:}B$ bipartition.
\end{corollary}

Using the relation in \eqref{eq:relation-unext-fid-sandwiched-1}--\eqref{eq:relation-unext-fid-sandwiched-last} 
between the unextendible fidelity and the $1/2$-sandwiched unextendible entanglement, 
we find that the unextendible fidelity of a pure bipartite state can be computed explicitly.

\begin{proposition}
Let $\psi_{AB}$ be a pure bipartite state. Then the following equality holds
\begin{equation}
F^u(\psi_{AB}) = \lambda_{\max}(\psi_A),
\end{equation}
where $\lambda_{\max}(\rho_A)$ is the maximal eigenvalue of $\psi_A$.
\end{proposition}

\section{Applications for secret key and entanglement distillation}
\label{sec:applications}

\subsection{Private states and unextendible entanglement}

In this section, we review the definition of a private state \cite{horodecki2005secure,horodecki2009general}, and then we
establish a bound on the number of private bits contained in a private state, in terms of the state's unextendible entanglement.
These results find applications in the next two subsections (Sections~\ref{sec:key-dist-overhead} and \ref{sec:exact-key-dist}),
where we investigate the overhead of probabilistic secret key distillation and the rate of exact secret key distillation.

We first review the definition of a private state~\cite{horodecki2005secure,horodecki2009general}.
Let $\rho_{ABA'B'}\in\cS({ABA'B'})$ be a state shared between spatially separated parties Alice and Bob,
where Alice possesses systems~$A$ and $A'$ and Bob possesses systems $B$ and $B'$, such
that
\begin{equation}
K\equiv\dim(\cH_A)=\dim(\cH_B).
\end{equation}
A state $\rho_{ABA'B'}$ is called a private state~\cite{horodecki2005secure,horodecki2009general} if Alice and Bob can
extract a secret key from it by performing local measurements on $A$ and $B$, such that the key is product with an arbitrary purifying
system of $\rho_{ABA'B'}$. That is, $\rho_{ABA'B'}$ is a private state of $\log  K$ private bits if, for every
purification $\varphi^\rho_{ABA'B'E}$ of $\rho_{ABA'B'}$, the following holds:
\begin{align}
    \left(\cM_A\ox\cM_B\ox\tr_{A'B'}\right)(\varphi^\rho_{ABA'B'E})
  = \frac{1}{K}\sum_{k=1}^K \proj{k}_A\ox\proj{k}_B\ox\sigma_E,
\end{align}
where $\cM(\cdot) = \sum_k\proj{k}(\cdot)\proj{k}$ is a projective measurement channel and $\sigma_E$ is some state on
the purifying system. The systems $A$ and $B$ are known as key systems, and $A'$ and $B'$ are known as shield systems. Interestingly, it was shown that a private state of $\log K$ private bits can be written in the
following form~\cite{horodecki2005secure,horodecki2009general}
\begin{align}
\gamma_{ABA^{\prime}B^{\prime}}\coloneqq U_{ABA^{\prime}B^{\prime}}
\left(\Phi^K_{AB}\otimes\sigma_{A^{\prime}B^{\prime}}\right)U_{ABA^{\prime}B^{\prime}%
}^{\dag},\label{eq:gamma-priv-state}%
\end{align}
where $\Phi^K_{AB}$ is a maximally entangled state of Schmidt rank $K$, the state $\sigma_{A^{\prime}%
B^{\prime}}$ is an arbitrary state, and
\begin{align}
U_{ABA^{\prime}B^{\prime}}\coloneqq \sum_{i,j=1}^K \proj{i}_A\otimes\proj{j}_B\otimes U_{A^{\prime}B^{\prime}}^{ij}
\end{align}
is a controlled unitary known as a ``twisting unitary,'' with each $U^{ij}_{A'B'}$ a unitary.

We now establish some bounds on the number of private bits contained in a private state, in terms of its unextendible entanglement.
Recall that the generalized mutual information of a state~$\rho_{AB}$ is defined as follows \cite{SW12}:
\begin{equation}
I_{\mathbf{D}}(A;B)_{\rho}\coloneqq \inf_{\sigma_{B}}\mathbf{D}(\rho_{AB}\Vert\rho
_{A}\otimes\sigma_{B}),
\end{equation}
where $\mathbf{D}$ is the generalized divergence discussed in Section~\ref{sec:generalized divergence} and the infimum is with respect to every density operator $\sigma_{B}$.
We first show that the unextendible entanglement of a private state is bounded from below by 
the generalized mutual information of $\Phi_{AB}$ (proof in Appendix~\ref{appx:lower-bound}). 

\begin{proposition}
\label{prop:lower-bound}
For a $\gamma$-bipartite private state of the form in
\eqref{eq:gamma-priv-state}, the following bound holds%
\begin{equation}
\mathbf{E}^{u}(\gamma_{AA^{\prime}BB^{\prime}})\geq\frac{1}{2}I_{\mathbf{D}%
}(A;B)_{\Phi^K},\label{eq:lower-bnd-priv-state}%
\end{equation}
where $I_{\mathbf{D}%
}(A;B)_{\Phi}$ is evaluated with respect to the state $\Phi^K_{AB}$ in \eqref{eq:gamma-priv-state}.
\end{proposition}

\bigskip
As a corollary of Proposition~\ref{prop:lower-bound} and Lemmas~\ref{lem:Petz-Renyi-MI-pure}, \ref{lem:sandwiched-Renyi-MI-pure}, and \ref{lem:geometric-Renyi-MI-pure}, we find that if the generalized divergence is set to be the Petz--, sandwiched, or geometric R\'enyi relative entropy,
then the unextendible entanglement of a $\gamma$-bipartite private state is bounded 
from below by the amount of secret key that can be extracted from the state.
Note that this corollary includes the relative entropy, the min-relative entropy, 
and the max-relative entropy as limiting cases.

\begin{corollary}
\label{cor:lower-bnd-unext-alpha-private}
If the generalized divergence is the Petz--, sandwiched, or geometric R\'enyi relative entropy with $\alpha$ set so that the data-processing inequality is satisfied, then the following bound holds%
\begin{equation}
\mathbf{E}_{\alpha}^u(\gamma_{AA^{\prime}BB^{\prime}})\geq\log  K.
\end{equation}
\end{corollary}

\subsection{Overhead of probabilistic secret key distillation}

\label{sec:key-dist-overhead}

This section considers the overhead of
probabilistic secret key distillation under selective
two-extendible operations (which include 1-LOCC operations as a special case).
We first formally define the overhead of probabilistic secret key distillation:
\begin{definition}
The overhead of distilling $k$ private bits from several independent copies of a bipartite state $\rho_{AB}$ is given
by the minimum number of copies of $\rho_{AB}$ needed, on average, to produce some private state $\gamma_{ABA'B'}^k$
with $\log  K = k$ using 1-LOCC operations:
\begin{align}
K_{\operatorname{ov}}(\rho_{AB}, k) \coloneqq
\inf\left\{\frac{n}{p}: \Lambda(\rho_{AB}^{\ox n})\to \gamma^k_{ABA'B'} \text{ with prob. } p,\
\Lambda\text{~is selective 1-LOCC}\right\}.
\end{align}
We can also define the overhead when selective two-extendible operations are allowed:
\begin{align}
K_{\operatorname{ov},2}(\rho_{AB}, k) \coloneqq
\inf\left\{\frac{n}{p}: \Lambda(\rho_{AB}^{\ox n})\to \gamma^k_{ABA'B'} \text{ with prob. } p,\
\Lambda\text{~is selective two-extendible}\right\}.
\end{align}
Note that it is not necessary to produce a particular private state $\gamma^k_{ABA'B'}$, but rather just some private state $\gamma^k_{ABA'B'}$ having $k$ private bits.
\end{definition}

The following inequality is a trivial consequence of definitions and the fact that a selective 1-LOCC operation is a special kind of selective two-extendible operation:
\begin{equation}
K_{\operatorname{ov}}(\rho_{AB}, k) \geq K_{\operatorname{ov},2}(\rho_{AB}, k).
\end{equation}

It turns out that the relative-entropy-induced unextendible entanglement $E^u(\rho_{AB})$,
given in Definition~\ref{def:rel-unextendible}, provides a lower
bound on the overhead of probabilistic secret key distillation of a 
bipartite state~$\rho_{AB}$ (proof in Appendix~\ref{appx:overhead-bnd-e-max-private}).

\begin{theorem}
\label{thm:overhead-bnd-e-max-private}
For a bipartite state $\rho_{AB}$, the overhead of distilling $k$ private bits from $\rho_{AB}$ is bounded from below by
\begin{align}
\label{eq:lower-bound-private}
K_{\operatorname{ov},2}(\rho_{AB}, k) \ge \frac{k}{E^u(\rho_{AB})},
\end{align}
where the relative-entropy-induced unextendible entanglement is given in Definition~\ref{def:rel-unextendible}.
\end{theorem}

In Theorem~\ref{thm:strong-2-ext-mono} and Proposition~\ref{prop:subadditivity},
we proved that the $\alpha$-Petz unextendible entanglement for
$\alpha\in[1,2]$, the $\alpha$-sandwiched unextendible entanglement for $\alpha\in[1,\infty]$, and the
$\alpha$-geometric unextendible entanglement for $\alpha\in[1,2]$ satisfy selective two-extendible monotonicity
and subadditivity, respectively. That is to say, all these entanglement measures can be used as lower
bounds in~\eqref{eq:lower-bound-private}. Since these divergences are monotonically increasing in
$\alpha$~\cite{tomamichel2015quantum} and $D(\omega\Vert\tau)\leq\widehat{D}(\omega\Vert\tau)$~\cite{hiai1991},
$E^{u}(\rho_{AB})$ is the smallest unextendible entanglement measure among these choices and yields the tightest
lower bound. Furthermore, the authors of~\cite{fawzi2019semidefinite} proposed a method to
accurately approximate the quantum relative entropy via semidefinite programming. This enables us to estimate the lower
bound using available semidefinite programming solvers. See also \cite{Fawzi_2018} in this context.

\subsection{Exact secret key distillation}

\label{sec:exact-key-dist}

We can also consider the setting in which the goal is to distill secret key exactly from a bipartite state
by using two-extendible or 1-LOCC operations. Though exact distillation is less realistic than the above
probabilistic scenario, it is still a core part of zero-error quantum information theory~\cite{GAM16}.

The one-shot 1-LOCC exact distillable key of $\rho_{AB}$ is
defined to be the maximum number of private bits achievable via a 1-LOCC channel; that is,
\begin{align}
K_{\operatorname{1-LOCC}}^{(1)}(\rho_{AB})
\coloneqq   \sup \{ k: \Lambda_{AB\to \hat{A}\hat{B}A'B'}(\rho_{AB}) = \gamma^k_{\hat{A}\hat{B}A'B'}, \ \Lambda \in \text{1-LOCC} \},
\end{align}
where $\gamma^k_{\hat{A}\hat{B}A'B'}$ is any private state with $k$ private bits. The 1-LOCC exact distillable key of a state~$\rho_{AB}$ is then defined as the regularization of $K_{\operatorname{1-LOCC}}^{(1)}(\rho_{AB})$:
\begin{align}
    K_{\operatorname{1-LOCC}}(\rho_{AB})
\coloneqq   \liminf_{n\to \infty} \frac1n K_{\operatorname{1-LOCC}}^{(1)}(\rho_{AB}^{\ox n}).
\end{align}
Note that the following inequality holds as a direct consequence of definitions:
\begin{equation}
K_{\operatorname{1-LOCC}}^{(1)}(\rho_{AB}) \leq
K_{\operatorname{1-LOCC}}(\rho_{AB}) .
\end{equation}

The one-shot two-extendible exact distillable key and two-extendible exact distillable key of~$\rho_{AB}$ are defined similarly:
\begin{align}
K_{\operatorname{2-EXT}}^{(1)}(\rho_{AB})
& \coloneqq   \sup \{ k: \Lambda_{AB\to \hat{A}\hat{B}A'B'}(\rho_{AB}) = \gamma^k_{\hat{A}\hat{B}A'B'}, \ \Lambda \in \text{2-EXT} \},\\
    K_{\operatorname{2-EXT}}(\rho_{AB})
 & \coloneqq   \liminf_{n\to \infty} \frac1n K_{\operatorname{2-EXT}}^{(1)}(\rho_{AB}^{\ox n}),
\end{align}
and the following inequality holds as a direct consequence of definitions:
\begin{equation}
K_{\operatorname{2-EXT}}^{(1)}(\rho_{AB}) \leq
K_{\operatorname{2-EXT}}(\rho_{AB}) .
\end{equation}

Immediate consequences of definitions and the fact that 1-LOCC operations are contained in the set of two-extendible operations are the following inequalities:
\begin{equation}
K_{\operatorname{1-LOCC}}^{(1)}(\rho_{AB}) \leq K_{\operatorname{2-EXT}}^{(1)}(\rho_{AB}),
\qquad
K_{\operatorname{1-LOCC}}(\rho_{AB}) \leq K_{\operatorname{2-EXT}}(\rho_{AB}),
\end{equation}

It turns out that the min-unextendible entanglement $E_{\min}^u(\rho_{AB})$ serves as an upper
bound on the two-extendible exact distillable key of a bipartite state~$\rho_{AB}$
(proof in Appendix~\ref{appx:exact-KD}).

\begin{theorem}\label{thm:exact-KD}
For a bipartite state $\rho_{AB}$, its asymptotic exact distillable key under two-extendible 
operations is bounded from above as
\begin{equation}
K_{\operatorname{2-EXT}}(\rho_{AB})  \le E_{\min}^u(\rho_{AB}).
\end{equation}
\end{theorem}

Actually, it is evident from the proof that in Theorem~\ref{thm:exact-KD}, 
both $E_{\max}^u(\rho_{AB})$ and $\widetilde{E}_{1/2}^u(\rho_{AB})$
are valid upper bounds on $K_{\operatorname{2-EXT}}(\rho_{AB})$, since they are also
monotonic and additive. However, since
\begin{align}\label{eq:inequality-chain}
		E_{\min}(\rho_{AB})
\leq \widetilde{E}_{1/2}^u(\rho_{AB})
\leq E_{\max}^u(\rho_{AB}),
\end{align}
the min-unextendible entanglement $E_{\min}^u(\rho_{AB})$ leads to the tightest upper bound.
Note that the inequalities in~\eqref{eq:inequality-chain} follow from the facts that
the fidelity obeys data processing under the channel $\mathcal{M}(\omega) \coloneqq \operatorname{Tr}[\Pi^\rho\omega] |0\rangle\!\langle 0| + \operatorname{Tr}[(\1-\Pi^\rho)\omega] |1\rangle\!\langle 1|$, so that $\operatorname{Tr}[\Pi^\rho \sigma] = F^2(\mathcal{M}(\rho),\mathcal{M}(\sigma)) \geq F^2(\rho,\sigma)$
and $\widetilde{D}_{\alpha}$ is monotonically increasing in $\alpha$~\cite{muller2013quantum,wilde2014strong}.

\subsection{Overhead of probabilistic entanglement distillation}

\label{sec:overhead-prob-ent-dist}

Entanglement distillation aims at obtaining maximally entangled states from less entangled bipartite states shared between two parties via certain free operations. As a central task in quantum information processing, various approaches \cite{Vedral1998,Rains1999,Vidal2002,Rains2001,Horodecki2000a,Christandl2004,Wang2016c,Leditzky2017,Wang2016m,Fang2017,Kaur2018,Regula2019} have been developed to characterize and
approximate the performance of the rates of deterministic entanglement distillation.

Here, we consider the overhead of
probabilistic entanglement
distillation~\cite{Bennett1996c,Cabrillo1999,Pan2001,Nickerson2014a,Campbell2008,Rozpedek2018} under selective
two-extendible operations, similar to what we considered for probabilistic secret key distillation.
We begin by defining the overhead of probabilistic entanglement distillation.

\begin{definition}
The overhead of distilling $m$ ebits from several independent copies of a bipartite state~$\rho_{AB}$ is equal to the
minimum number of copies of $\rho_{AB}$ needed, on average, to produce $m$~copies of
the ebit $\Phi^2$ using selective 1-LOCC
operations:
\begin{equation}
E_{\operatorname{ov}}(\rho_{AB}, m) \coloneqq
\inf\left\{\frac{n}{p}: \Lambda(\rho_{AB}^{\ox n})\to (\Phi^2)^{\ox m} \text{ with prob. } p,\
\Lambda\text{~is selective 1-LOCC}\right\}.
\end{equation}
We can also define the overhead when selective two-extendible operations are allowed:
\begin{equation}
E_{\operatorname{ov},2}(\rho_{AB}, m) \coloneqq
\inf\left\{\frac{n}{p}: \Lambda(\rho_{AB}^{\ox n})\to (\Phi^2)^{\ox m} \text{ with prob. } p,\
\Lambda\text{~is selective two-extendible}\right\}.
\end{equation}
\end{definition}

The following inequality is a trivial consequence of definitions and the fact that a selective 1-LOCC operation is a special kind
of selective two-extendible operation:
\begin{equation}
E_{\operatorname{ov}}(\rho_{AB}, m) \geq E_{\operatorname{ov},2}(\rho_{AB}, m).
\label{eq:ov-1-locc-2-ext-rel}
\end{equation}

Since secret key can be obtained from ebits~\cite{horodecki2005secure} via local operations,
a direct corollary of Theorem~\ref{thm:overhead-bnd-e-max-private}
is that $E^{u}$ is a lower bound on the overhead of probabilistic entanglement distillation.

\begin{corollary}\label{thm:overhead-bnd-e-max}
For a bipartite state $\rho_{AB}$, the overhead of distilling $m$ ebits from $\rho_{AB}$ is bounded from below as
\begin{align}\label{eq:lower-bound}
E_{\operatorname{ov},2}(\rho_{AB}, m) \ge \frac{m}{E^{u}(\rho_{AB})},
\end{align}
where $E^{u}(\rho_{AB})$ is the relative-entropy-induced unextendible entanglement
from Definition~\ref{def:rel-unextendible}.
\end{corollary}

\subsection{Exact entanglement distillation}

The one-shot 1-LOCC exact distillable entanglement of $\rho_{AB}$ is
defined to be the maximum number of ebits achievable via a 1-LOCC channel; that is,
\begin{align}
E_{\operatorname{1-LOCC}}^{(1)}(\rho_{AB})
\coloneqq   \sup \{ \log  d: \Lambda_{AB\to \hat{A}\hat{B}}(\rho_{AB}) = \Phi^d_{\hat{A}\hat{B}}, \ \Lambda \in \text{1-LOCC} \},
\end{align}
where $\Phi^d_{\hat{A}\hat{B}}$ is a maximally entangled state of Schmidt rank $d$.
The 1-LOCC exact distillable entanglement of a state
$\rho_{AB}$ is then defined as the regularization of $E_{\operatorname{1-LOCC}}^{(1)}(\rho_{AB})$:
\begin{align}
    E_{\operatorname{1-LOCC}}(\rho_{AB})
\coloneqq   \liminf_{n\to \infty} \frac1n E_{\operatorname{1-LOCC}}^{(1)}(\rho_{AB}^{\ox n}).
\end{align}
Note that the following inequality is a direct consequence of the definitions:
\begin{equation}
E_{\operatorname{1-LOCC}}^{(1)}(\rho_{AB}) \leq
E_{\operatorname{1-LOCC}}(\rho_{AB}).
\end{equation}

The one-shot two-extendible exact distillable entanglement and two-extendible exact distillable entanglement of $\rho_{AB}$ are defined similarly:
\begin{align}
E_{\operatorname{2-EXT}}^{(1)}(\rho_{AB})
& \coloneqq   \sup \{ \log  d: \Lambda_{AB\to \hat{A}\hat{B}}(\rho_{AB}) = \Phi^d_{\hat{A}\hat{B}}, \ \Lambda \in \text{2-EXT} \},\\
    E_{\operatorname{2-EXT}}(\rho_{AB})
 & \coloneqq   \liminf_{n\to \infty} \frac1n E_{\operatorname{2-EXT}}^{(1)}(\rho_{AB}^{\ox n}),
\end{align}
and the following inequality is a direct consequence of definitions:
\begin{equation}
E_{\operatorname{2-EXT}}^{(1)}(\rho_{AB}) \leq
E_{\operatorname{2-EXT}}(\rho_{AB}).
\end{equation}

Immediate consequences of definitions and the fact that 1-LOCC operations are contained in the set of two-extendible operations are the following inequalities:
\begin{equation}
E_{\operatorname{1-LOCC}}^{(1)}(\rho_{AB}) \leq E_{\operatorname{2-EXT}}^{(1)}(\rho_{AB}),
\qquad
E_{\operatorname{1-LOCC}}(\rho_{AB}) \leq E_{\operatorname{2-EXT}}(\rho_{AB}),
\end{equation}

As a direct corollary of Theorem~\ref{thm:exact-KD} and the relation between private states and ebits mentioned previously (just below \eqref{eq:ov-1-locc-2-ext-rel}),
the min-unextendible entanglement $E_{\min}^u$ serves as an upper
bound on the exact distillable entanglement via two-extendible operations.

\begin{corollary}\label{thm:exact-ED}
For a bipartite state $\rho_{AB}$, its asymptotic exact distillable entanglement under two-extendible operations is bounded from above by its min-unextendible entanglement:
\begin{equation}
E_{\operatorname{2-EXT}}(\rho_{AB}) \le E_{\min}^u(\rho_{AB}).
\end{equation}
\end{corollary}

\subsection{Examples}\label{sec:examples}

In this section, we apply our bounds on the overhead of probabilistic entanglement or secret key distillation to three classes of states: isotropic states, S~states, and erased states. We compare our lower bounds and other known estimations of the overhead to upper bounds derived from known distillation protocols. In particular, we show that our lower bound on the overhead of distillation is tight for erased states.

To begin with, recall the relative entropy of entanglement \cite{Vedral1998}:
\begin{equation}
E_R(\rho_{AB}) \coloneqq \min_{\sigma_{AB} \in \operatorname{SEP}} D(\rho_{AB}\Vert \sigma_{AB}),
\end{equation}
where SEP denotes the set of separable states. It is known that the relative entropy of entanglement is monotone under selective LOCC \cite{Vedral1998}. Thus, the relative entropy of entanglement $E_R(\rho_{AB})$ can be used to estimate the overhead of distillation under selective LOCC via an approach similar
to that given in the proof of Corollary~\ref{thm:overhead-bnd-e-max}.

\textbf{Isotropic states:} Let us first investigate the class of isotropic states $\rho_r$, defined as~\cite{PhysRevA.59.4206}
\begin{equation}
\rho_{r}\coloneqq r \cdot \Phi^d + (1-r) \frac{\1 -\Phi^d}{d^{2}-1},
\end{equation}
where $r \in[0,1]$, so that $\rho_r$ is a convex mixture of a maximally entangled state
$\Phi^d$ of Schmidt rank~$d$ and its orthogonal complement.
Numerous works have been carried out to study the distillation rate of isotropic states
under various scenarios~\cite{Fang2017,wang2018semidefinite,Leditzky2017,Rozpedek2018}.
Here we investigate probabilistic entanglement distillation under 1-LOCC operations and show the advantage of our method in estimating the overhead.
We consider the case $d=2$ for simplicity. Figure~\ref{fig:isotropic} plots the unextendible
entanglement and relative entropy of entanglement for the overhead of probabilistic entanglement or
secret key distillation for this set of states.\footnote{All Matlab codes used to generate these plots and calculate various unextendible entanglement measures are available with the arXiv posting of this paper as arXiv ancillary files.} Note that the relative entropy of entanglement for qubit-qubit isotropic states ($d=2$) was calculated in \cite{VPRK97}.

\begin{figure}
    \centering
    \includegraphics[width=7.5cm]{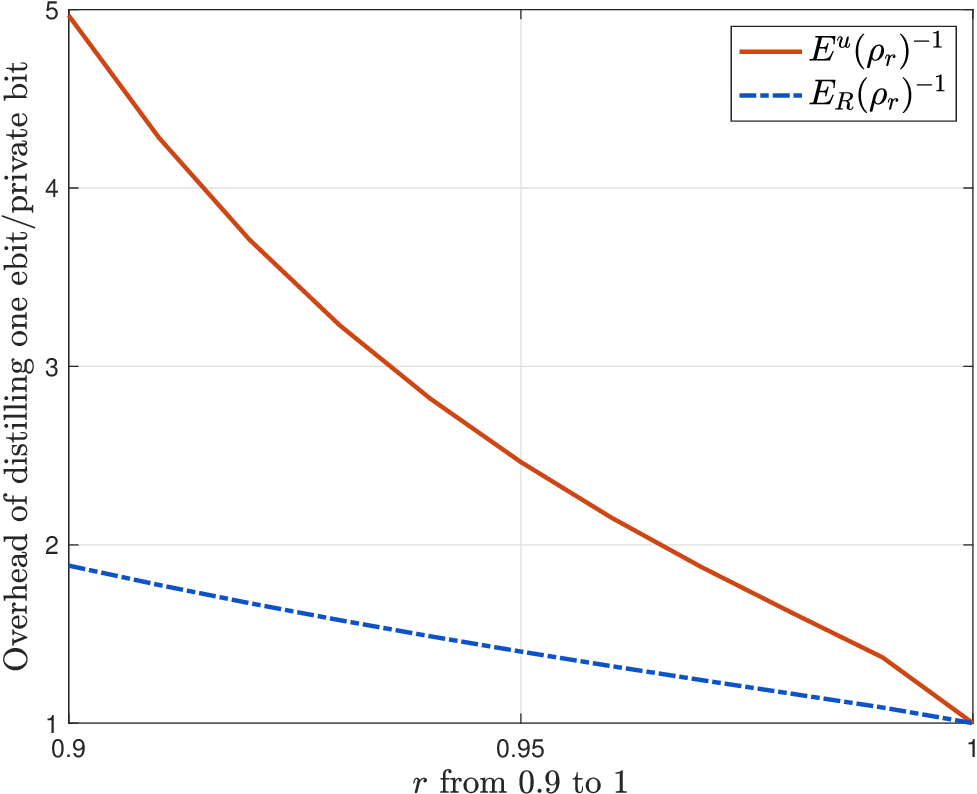}
    \caption{Lower bounds on the overhead of distilling one ebit or one private bit from isotropic states. For the overhead of probabilistic distillation of entanglement or secret key under 1-LOCC operations, our lower bound (solid line) outperforms the previous lower bound (dashed line) based on relative entropy of entanglement.}
    \label{fig:isotropic}
\end{figure}

\textbf{S~states:} An S~state is a mixture of the Bell state and non-orthogonal product noise, whose distillation protocols were studied in \cite{Rozpedek2018,Zhao2021}. Here, we define it as
\begin{equation}
S_p = p \Phi^2 + (1-p)\ketbra{00}{00}.
\end{equation}
Figure~\ref{fig:S state} plots the unextendible
entanglement and relative entropy of entanglement for the overhead of probabilistic entanglement or
secret key distillation for this set of states. The result shows that our result establishes a better lower bound on the overhead of probabilistic distillation under 1-LOCC operations.

\begin{figure}
    \centering
    \includegraphics[width=7.5cm]{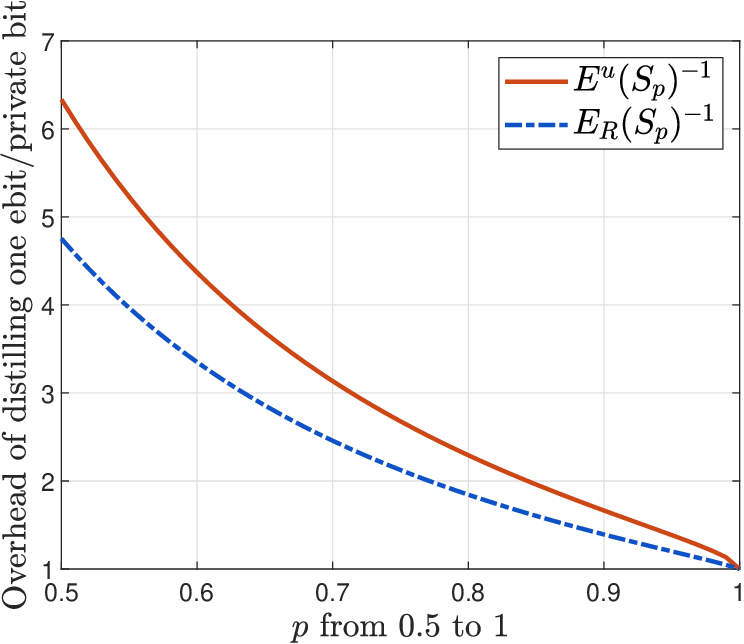}
    \caption{Bounds on the overhead of distilling one ebit or one private bit from S~states.  For the overhead of probabilistic distillation of entanglement or secret key under 1-LOCC operations, our lower bound (solid line) outperforms the previous lower bound (dashed line) based on relative entropy of entanglement.}
    \label{fig:S state}
\end{figure}

\textbf{Erased states:} We also consider the class of erased states, which are the Choi states of quantum
erasure channels \cite{GBP97}:
\begin{align}
    \rho_{A'B}^\varepsilon = (1 - \varepsilon)\Phi^2_{AB} + \varepsilon\proj{e}_{A'}\ox\pi_B,
\end{align}
where $\varepsilon\in[0,1]$, $\Phi^2_{AB}$ is an ebit, and $\ket{e}$ is some state that is
orthogonal to $\pi_A$. For simplicity, we choose $d_A = d_B = 2$ (qubit system), and $d_{A'} = d_A + 1 = 3$.
This state can be obtained as follows.
Alice and Bob share a two-qubit maximally entangled state $\Phi_{AB}$ and
Alice transmits her local copy through an erasure channel, so that with probability $1-\varepsilon$ it is unaffected
and with probability $\varepsilon$ it is ``erased'' (replaced with the erasure flag state $e$) .

Interestingly, there is a simple formula for the relative-entropy-induced 
unextendible entanglement of the erased state (proof in Appendix~\ref{appx:unext-erased-state}).

\begin{proposition}\label{prop:unext-erased-state}
The unextendible entanglement of the erased state $\rho_{A'B}^\varepsilon$ is $E^{u}(\rho_{A'B}^\varepsilon) = 1 - \varepsilon$.
\end{proposition}

\bigskip
Corollary~\ref{thm:overhead-bnd-e-max} together with Proposition~\ref{prop:unext-erased-state}
shows that $1/E^u(\rho^\varepsilon_{A'B})= 1/(1-\varepsilon)$ is a
lower bound on the overhead of probabilistically distilling one ebit from $\rho_{A'B}^\varepsilon$.
Indeed, we can show the tightness of the lower bound by considering the following
one-way LOCC protocol that achieves this bound. 
Given $\rho_{A'B}^\varepsilon$, Alice detects whether her local
system is erased or not by performing a binary measurement $\{\1_A, \proj{e}\}$. 
With probability $1-\varepsilon$, she
finds that the system is not erased. Then she sends this information to Bob, and the protocol 
finishes with one ebit shared between them.
As such, we conclude the following.

\begin{theorem}
For the erased state $\rho^\varepsilon_{A'B}$, we have
\begin{align}
E_{\operatorname{ov}}(\rho^\varepsilon_{A'B}, 1) = E_{\operatorname{ov},2}(\rho^\varepsilon_{A'B}, 1) = \frac{1}{E^u(\rho^\varepsilon_{A'B})} = \frac{1}{1-\varepsilon}.
\end{align}
\end{theorem}

\bigskip
Therefore, our estimation on the overhead of distilling ebits from erased states under selective one-way LOCC is
optimal in the sense that the upper bound on the  overhead of the above protocol matches our lower bound from
Corollary~\ref{thm:overhead-bnd-e-max}, thus
characterizing the ability of probabilistic distillation for erased states.
Note that this result can be generalized to multiple copies of the erased state.

We further consider the overhead of distillation under LOCC operations for the erased states 
in the following proposition (proof in Appendix~\ref{appx:erasure-state-bounds}), 
indicating an interesting fact that the optimal protocol for
distilling one ebit from erased states under one-way LOCC operations matches the lower bound
$1/E_R(\rho_{A'B}^\varepsilon)\equiv1/(1-\varepsilon)$ on the overhead under LOCC operations.
It reveals that one-way LOCC operations have the
same power and performance as LOCC operations in probabilistically distilling ebits from erased states.

\begin{proposition}
\label{prop:erasure-state-bounds}
The relative entropy of entanglement of the erased state $\rho_{A'B}^\varepsilon$ is $E_R(\rho_{A'B}^\varepsilon) = 1 - \varepsilon$, which implies that the overhead of distillation under selective LOCC for $\rho_{A'B}^\varepsilon$ is given by $(1-\varepsilon)^{-1}$.
\end{proposition}

\section{Concluding remarks}

Our work introduces a family of entanglement measures called unextendible entanglement to quantify and investigate the unextendibility of entanglement. The crucial technical contribution was intuitively motivated by the fact that the more entangled a bipartite state is, the less each of its constituent systems can be entangled with a third party. A key distinction to previous works is that our proposed measures restrict the free set to be state dependent. These entanglement measures have desirable properties, including monotonicity under selective two-extendible operations, faithfulness, and normalization. The unextendible entanglement in our work has also found direct operational applications in evaluating the overhead and rate of entanglement or secret key distillation. As a notable application example, we characterized the optimal overhead of distilling one ebit from the erased states under one-way LOCC operations. We also note that the probabilistic state conversion task, which has already been studied for the pure-state case~\cite{vidal1999entanglement,alhambra2016fluctuating}, can be further explored using the results in this paper.

An important problem for future work is to investigate to what extent the bounds on the asymptotic distillable secret
key or entanglement can be approached. One potential approach is to explore the connections to distinguishability distillation (i.e., hypothesis testing)~\cite{Nagaoka2006,Hayashi2007,Audenaert2008a,Mosonyi2015,WW19}. 
It is also of interest to consider an extension to the resource theory of
$k$-unextendibility and entanglement dilution~\cite{Bennett1996c,Hayden2001,Vidal2002b,Wang2016d,Buscemi2011,WW18},
where the techniques applied in~\cite{xie2017approximate} might be useful here.
Moreover, we have only considered the extendibility on system $B$ when defining generalized unextendible 
entanglement measure, yielding an asymmetric entanglement measure.
It would be meaningful to further explore the extendibility on system $A$ and promote 
our measure to a symmetric one. 
Another interesting direction is to
develop the continuous-variable setting, where the results of \cite{LKAW19} could be helpful, 
as well as the dynamical setting of quantum channels.

\section*{Acknowledgement}

The authors thank Anurag Anshu, Kun Fang,  Debbie Leung, Vishal Singh, and Dave Touchette for discussions.
Part of this work was done when
KW was at the Southern University of Science and Technology and XW was at the University of Maryland.
XW was partially supported by the Start-up Fund from The Hong Kong University of Science and Technology (Guangzhou), 
the Guangdong Quantum Science and Technology Strategic Fund (Grant No. GDZX2303007), 
and the Education Bureau of Guangzhou Municipality.
MMW acknowledges support from NSF under grant no.~2315398.

\bibliographystyle{alpha}
\bibliography{Ref.bib}

\appendix

\section{Proof of Lemma~\ref{lem:faithfulness-renyi-rel-ents}}
\label{appx:faithfulness-renyi-rel-ents}

By utilizing the data-processing inequality and the faithfulness of the classical R\'enyi relative entropy, it follows that the statements above are true for all of the quantum R\'enyi relative entropies for the range of $\alpha$ for which data processing holds (see, e.g., \cite{wilde2014strong} for this kind of argument). To get outside of the range for $\alpha \in (0,1)$, we use the following inequality holding for $\alpha \in (0,1)$ \cite{IRS17}:
\begin{equation}
\alpha D_\alpha(\omega \Vert \tau) \leq \widetilde{D}_{\alpha}(\omega\Vert\tau).
\end{equation}
To get outside of the range for $\alpha \in (1,\infty)$, we use the following inequality holding for $\alpha \in (1,\infty)$:
\begin{equation}
\widetilde{D}_{\alpha}(\omega\Vert\tau) \leq D_{\alpha}(\omega\Vert\tau) .
\end{equation}
This inequality was proved in \cite{wilde2014strong,DL14limit} (following from the Araki--Lieb--Thirring inequality \cite{Araki1990,LT76}).
For the same range, we also make use of the following inequality:
\begin{equation}
\widetilde{D}_{\alpha}(\omega\Vert\tau)
\leq \widehat{D}_{\alpha}(\omega\Vert\tau).
\end{equation}
which also follows from the Araki--Lieb--Thirring inequality for $\alpha \in (1,\infty)$ because
\begin{align}
\widetilde{Q}_\alpha(\omega \Vert \tau) & = \tr[(\tau^{(1-\alpha)/2\alpha} \omega \tau^{(1-\alpha)/2\alpha})^{\alpha}]\\
& =
\tr[(\tau^{1/2\alpha} \tau^{-1/2} \omega \tau^{-1/2} \tau^{1/2\alpha} )^{\alpha}]\\
& \leq
\tr[(\tau^{1/2\alpha})^{\alpha} (\tau^{-1/2} \omega \tau^{-1/2})^{\alpha} (\tau^{1/2\alpha})^{\alpha}]\\
& = \tr[\tau (\tau^{-1/2} \omega \tau^{-1/2})^{\alpha} ]
\\
& = \widehat{Q}_\alpha(\omega \Vert \tau).
\end{align}
This concludes the proof.

\section{Proof of Theorem~\ref{thm:strong-2-ext-mono}}\label{appx:strong-2-ext-mono}

We first show the proof for \eqref{eq:full-monotone-ineq-sandwiched}.
Let $\rho_{AB_{1}B_{2}}$ be an arbitrary extension of $\rho_{AB_{1}}$. Then
define%
\begin{align}
p(y)  &  \coloneqq \operatorname{Tr}[\mathcal{E}_{AB_{1}B_{2}\rightarrow A^{\prime
}B_{1}^{\prime}B_{2}^{\prime}}^{y}(\rho_{AB_{1}B_{2}})],\label{eq:p(y)-pf}\\
\omega_{A^{\prime}B_{1}^{\prime}B_{2}^{\prime}}^{y}  &  \coloneqq \frac{1}%
{p(y)}\mathcal{E}_{AB_{1}B_{2}\rightarrow A^{\prime}B_{1}^{\prime}%
B_{2}^{\prime}}^{y}(\rho_{AB_{1}B_{2}}),\label{eq:omega-proof}
\end{align}
where $\mathcal{E}_{AB_{1}B_{2}\rightarrow A^{\prime}B_{1}^{\prime}%
B_{2}^{\prime}}^{y}$ is the operation extending $\mathcal{E}_{AB\rightarrow
A^{\prime}B^{\prime}}^{y}$. Due to the channel extension property in~\eqref{eq:marginal},
the values of $p(y)$ in \eqref{eq:p(y)-thm} and \eqref{eq:p(y)-pf} are equal
and $\omega_{A^{\prime}B_{1}^{\prime}B_{2}^{\prime}}^{y}$ defined in~\eqref{eq:omega-proof}
is an extension of $\omega_{A^{\prime}B^{\prime}}^{y}$, the latter defined in~\eqref{eq:omega-thm} for arbitrary $y$.

From the data-processing inequality for the sandwiched R\'enyi relative entropy
for $\alpha\in\lbrack1/2,1)\cup(1,\infty)$, we conclude the following
inequality:%
\begin{equation}
\widetilde{D}_{\alpha}(\rho_{AB_{1}}\Vert\rho_{AB_{2}})\geq\widetilde
{D}_{\alpha}(\mathcal{E}_{AB\rightarrow A^{\prime}B^{\prime}Y}(\rho_{AB_{1}%
})\Vert\mathcal{E}_{AB\rightarrow A^{\prime}B^{\prime}Y}(\rho_{AB_{2}})),
\end{equation}
where the quantum channel $\mathcal{E}_{AB\rightarrow A^{\prime}B^{\prime}Y}$
is defined as%
\begin{equation}
\mathcal{E}_{AB\rightarrow A^{\prime}B^{\prime}Y}(\cdot)\coloneqq \sum_{y}%
\mathcal{E}_{AB\rightarrow A^{\prime}B^{\prime}}^{y}(\cdot)\otimes
|y\rangle\!\langle y|_{Y}.
\end{equation}
Since the marginal operations $\mathcal{E}_{AB_{1}\rightarrow A^{\prime}%
B_{1}^{\prime}}^{y}$ and $\mathcal{E}_{AB_{2}\rightarrow A^{\prime}%
B_{2}^{\prime}}^{y}$ are each equal to the original operation $\mathcal{E}%
_{AB\rightarrow A^{\prime}B^{\prime}}^{y}$, we find that%
\begin{align}
\mathcal{E}_{AB\rightarrow A^{\prime}B^{\prime}Y}(\rho_{AB_{1}}) &  =\sum
_{y}\mathcal{E}_{AB_{1}\rightarrow A^{\prime}B_{1}^{\prime}}^{y}(\rho_{AB_{1}%
})\otimes|y\rangle\!\langle y|_{Y}\\
&  =\sum_{y}p(y)\omega_{A^{\prime}B_{1}^{\prime}}^{y}\otimes|y\rangle\!\langle
y|_{Y}\\
&  \eqqcolon \omega_{A^{\prime}B_{1}^{\prime}Y},\\
\mathcal{E}_{AB\rightarrow A^{\prime}B^{\prime}Y}(\rho_{AB_{2}}) &  =\sum
_{y}\mathcal{E}_{AB_{2}\rightarrow A^{\prime}B_{2}^{\prime}}^{y}(\rho_{AB_{2}%
})\otimes|y\rangle\!\langle y|_{Y}\\
&  =\sum_{y}p(y)\omega_{A^{\prime}B_{2}^{\prime}}^{y}\otimes|y\rangle\!\langle
y|_{Y}\\
&  \eqqcolon \omega_{A^{\prime}B_{2}^{\prime}Y},
\end{align}
which implies that%
\begin{align}
  \widetilde{D}_{\alpha}(\mathcal{E}_{AB\rightarrow A^{\prime}B^{\prime}%
Y}(\rho_{AB_{1}})\Vert\mathcal{E}_{AB\rightarrow A^{\prime}B^{\prime}Y}%
(\rho_{AB_{2}}))
& = \widetilde{D}_{\alpha}(\omega_{A^{\prime}B_{1}^{\prime}Y}\Vert
\omega_{A^{\prime}B_{2}^{\prime}Y})\\
& = \frac{1}{\alpha-1}\log\sum_{y:p(y)>0}p(y)\widetilde{Q}_{\alpha}%
(\omega_{A^{\prime}B_{1}^{\prime}}^{y}\Vert\omega_{A^{\prime}B_{2}^{\prime}}^{y})\\
& \geq \sum_{y:p(y)>0}p(y)\left[  \frac{1}{\alpha-1}\log\widetilde
{Q}_{\alpha}(\omega_{A^{\prime}B_{1}^{\prime}}^{y}\Vert\omega_{A^{\prime}%
B_{2}^{\prime}}^{y})\right]  \\
& = \sum_{y:p(y)>0}p(y)\widetilde{D}_{\alpha}(\omega_{A^{\prime}B_{1}^{\prime
}}^{y}\Vert\omega_{A^{\prime}B_{2}^{\prime}}^{y})\\
& \geq 2\sum_{y:p(y)>0}p(y)E_{\alpha}^{u}(\omega_{A^{\prime}B_{1}^{\prime}}^{y})\\
& = 2\sum_{y:p(y)>0}p(y)E_{\alpha}^{u}(\omega_{A^{\prime}B^{\prime}}^{y}).
\end{align}
In the above, the first inequality follows from concavity of the logarithm and the fact that $\alpha>1$. The second
inequality follows because $\omega _{A^{\prime}B_{1}^{\prime}B_{2}^{\prime}}^{y}$ is such that $\operatorname{Tr}%
_{B_{2}^{\prime}}[\omega_{A^{\prime}B_{1}^{\prime}B_{2}^{\prime}}^{y}%
]=\omega_{A^{\prime}B_{1}^{\prime}}^{y}$, so that we can then optimize over
every extension of $\omega_{A^{\prime}B_{1}^{\prime}}^{y}$.\ Since the
following inequality has been established for every extension $\rho
_{AB_{1}B_{2}}$ of $\rho_{AB}$:
\begin{equation}
\frac{1}{2}\widetilde{D}_{\alpha}(\rho_{AB_{1}}\Vert\rho_{AB_{2}})\geq
\sum_{y:p(y)>0}p(y)E_{\alpha}^{u}(\omega_{A^{\prime}B^{\prime}}^{y}),
\end{equation}
we conclude the inequality in \eqref{eq:full-monotone-ineq-sandwiched}.

The proof of \eqref{eq:full-monotone-ineq} follows along similar lines, but instead using data processing for the Petz--R\'enyi relative entropy and the quantity $Q_{\alpha}$ in \eqref{eq:quasi-Petz}.
In the same way, the proof of \eqref{eq:full-monotone-ineq-geometric} follows, but instead using data processing for the geometric R\'enyi relative entropy and the quantity $\widehat{Q}_{\alpha}$ in~\eqref{eq:quasi-geometric}.

The inequality in \eqref{eq:full-monotone-ineq-sandwiched} for $\alpha \in \{1,\infty\}$ follows by taking a limit. Alternatively, we can prove them directly.
To establish the inequality in \eqref{eq:full-monotone-ineq-sandwiched} for $\alpha=1$
(for quantum relative entropy), we use the same reasoning but the following
property of quantum relative entropy:%
\begin{equation}
D(\omega_{XB}\Vert\tau_{XB})=\sum_{x}p(x)D(\omega_{B}^{x}\Vert\tau_{B}%
^{x})+D(p\Vert q),
\end{equation}
where%
\begin{align}
\omega_{XB}  & \coloneqq \sum_{x}p(x)|x\rangle\!\langle x|_{X}\otimes\omega_{B}^{x},\\
\tau_{XB}  & \coloneqq \sum_{x}q(x)|x\rangle\!\langle x|_{X}\otimes\tau_{B}^{x}.
\end{align}
Similarly, the inequality in \eqref{eq:full-monotone-ineq-geometric} for $\alpha=1$ follows from the same reasoning and the following related property of Belavkin--Staszewski relative entropy:
\begin{equation}
\widehat{D}(\omega_{XB}\Vert\tau_{XB})=\sum_{x}p(x)\widehat{D}(\omega_{B}^{x}\Vert\tau_{B}%
^{x})+D(p\Vert q).
\end{equation}

For $\alpha = \infty$, suppose that the optimal extension state is $\rho_{AB_1B_2}$, which satisfies $\rho_{AB_1}=\rho_{AB}$ and $E^u_{\max}(\rho_{AB_1})=\frac12 D_{\max}(\rho_{AB_1}\|\rho_{AB_2})$. Also, suppose that $E^u_{\max}(\rho_{AB_1})=\frac12 \log  t$, which means that $\rho_{AB_1}\le t\rho_{AB_2}$ For every $y$ such that $p(y)>0$, due to complete positivity and two-extendibility of the map $\cE^{y}_{AB\rightarrow A^{\prime}B^{\prime}}$, we have that
\begin{align}
t\,\cE^{y}_{AB_2\rightarrow A^{\prime}B_2^{\prime}}(\rho_{AB_2})/p(y) \ge
 \cE^{y}_{AB_1\rightarrow A^{\prime}B_1^{\prime}}(\rho_{AB_1})/p(y)=\omega_{A'B'}^y.
\end{align}
Noting that $\cE^{y}_{AB_1B_2\rightarrow A^{\prime}B_1^{\prime}B_2^{\prime}}(\rho_{AB_1B_2})/p(y)$ is an extension of the state $\omega_{A'B'}^y$, it follows that
\begin{align}
\label{eq: omega y Emax}
E^u_{\max}(\omega_{A'B'}^y) \le \frac12 \log  t.
\end{align}
Therefore,
\begin{align}
      E^u_{\max}(\rho_{AB_1}) &= \frac{1}{2} \log  t\\
      & = \frac{1}{2} \log  \sum_{y : p(y)>0}p(y)t\\
      & \ge \frac{1}{2} \log  \sum_{y: p(y)>0}p(y) 2^{2E^u_{\max}(\omega_{A'B'}^y) } \label{eq: emax ineq 1}\\
      & \ge   \sum_{y: p(y)>0}p(y) {E^u_{\max}(\omega_{A'B'}^y) } \label{eq: emax ineq 2}
\end{align}
In the above, \eqref{eq: emax ineq 1} is due to \eqref{eq: omega y Emax}, while \eqref{eq: emax ineq 2} is due to the concavity of the logarithm.

\section{Proof of Lemma~\ref{lem:sandwiched-Renyi-MI-pure}}

\label{app:proof-lemma-sw-Rny-MI-pure}

Without loss of generality, let $\sigma_{B}$ be an arbitrary state with support equal to the support of the reduced state $\psi_B$. Then we have that
\begin{align}
  \widetilde{D}_{\alpha}(\psi_{AB}\Vert\psi_{A}\otimes\sigma_{B})
&  =\frac{1}{\alpha-1}\log\operatorname{Tr}\!\left[  \left(  \psi_{AB}%
^{1/2}\left(  \psi_{A}\otimes\sigma_{B}\right)  ^{\frac{1-\alpha}{\alpha}}%
\psi_{AB}^{1/2}\right)  ^{\alpha}\right]  \\
&  =\frac{1}{\alpha-1}\log\operatorname{Tr}\!\left[  \left(  \psi
_{AB}\left(  \psi_{A}\otimes\sigma_{B}\right)  ^{\frac{1-\alpha}{\alpha}}%
\psi_{AB}\right)  ^{\alpha}\right]  \\
&  =\frac{1}{\alpha-1}\log\operatorname{Tr}\!\left[  \left(  |\psi
\rangle\!\langle\psi|_{AB}\left(  \psi_{A}\otimes\sigma_{B}\right)
^{\frac{1-\alpha}{\alpha}}|\psi\rangle\!\langle\psi|_{AB}\right)  ^{\alpha
}\right]  \\
&  =\frac{1}{\alpha-1}\log\operatorname{Tr}\!\left[  \left(  \langle
\psi|_{AB}\left(  \psi_{A}\otimes\sigma_{B}\right)  ^{\frac{1-\alpha}{\alpha}%
}|\psi\rangle|\psi\rangle\!\langle\psi|_{AB}\right)  ^{\alpha}\right]  \\
&  =\frac{1}{\alpha-1}\log\operatorname{Tr}\!\left[  \left(
\operatorname{Tr}\!\left[  \psi_{AB}\left(  \psi_{A}\otimes\sigma_{B}\right)
^{\frac{1-\alpha}{\alpha}}\right]  |\psi\rangle\!\langle\psi|_{AB}\right)
^{\alpha}\right]  \\
&  =\frac{1}{\alpha-1}\log\left(  \operatorname{Tr}\!\left[  \psi
_{AB}\left(  \psi_{A}\otimes\sigma_{B}\right)  ^{\frac{1-\alpha}{\alpha}%
}\right]  \right)  ^{\alpha}\operatorname{Tr}\!\left[  |\psi\rangle\!\langle
\psi|_{AB}^{\alpha}\right]  \\
&  =\frac{\alpha}{\alpha-1}\log\operatorname{Tr}\!\left[  \psi_{AB}\left(
\psi_{A}\otimes\sigma_{B}\right)  ^{\frac{1-\alpha}{\alpha}}\right].
\end{align}
For every pure bipartite state $\psi_{AB}$, there exists an operator $X_{A}$
such that $\psi_{AB}=X_{A}\Gamma_{AB}X_{A}^{\dag}$ and $\operatorname{Tr}%
[X_{A}^{\dag}X_{A}]=1$, which implies that $\psi_{A}=X_{A}X_{A}^{\dag}$. Note that we define
\begin{equation}
\Gamma_{AB}\coloneqq |\Gamma\rangle\!\langle\Gamma|_{AB},\qquad|\Gamma\rangle_{AB}%
\coloneqq \sum_{i=1}^{d}|i\rangle_{A}|i\rangle_{B},
\end{equation}
where
$\left\{  |i\rangle_{A}\right\}  _{i}$ and $\left\{  |i\rangle_{B}\right\}
_{i}$ are orthonormal bases.
Furthermore, by taking the polar decomposition of $X_{A}$, there exists a
unitary $U_{A}$ and a density operator $\rho_{A}$ (having the same spectrum as $\psi_{A}$) such that
$X_{A}=U_{A}\sqrt{\rho_{A}}$. Then consider that the last line above is equal
to%
\begin{align}
&  \frac{\alpha}{\alpha-1}\log\operatorname{Tr}\!\left[  \psi_{AB}\left(
\psi_{A}\otimes\sigma_{B}\right)  ^{\frac{1-\alpha}{\alpha}}\right]
\nonumber\\
&  =\frac{\alpha}{\alpha-1}\log\operatorname{Tr}\!\left[  X_{A}\Gamma
_{AB}X_{A}^{\dag}\left(  X_{A}X_{A}^{\dag}\otimes\sigma_{B}\right)
^{\frac{1-\alpha}{\alpha}}\right]  \\
&  =\frac{\alpha}{\alpha-1}\log\operatorname{Tr}\!\left[  X_{A}\Gamma
_{AB}X_{A}^{\dag}\left(  \left(  X_{A}X_{A}^{\dag}\right)  ^{\frac{1-\alpha
}{\alpha}}\otimes\sigma_{B}^{\frac{1-\alpha}{\alpha}}\right)  \right]  \\
&  =\frac{\alpha}{\alpha-1}\log\operatorname{Tr}\!\left[  \Gamma_{AB}\left(
X_{A}^{\dag}\left(  X_{A}X_{A}^{\dag}\right)  ^{\frac{1-\alpha}{\alpha}}%
X_{A}\otimes\sigma_{B}^{\frac{1-\alpha}{\alpha}}\right)  \right]  \\
&  =\frac{\alpha}{\alpha-1}\log\operatorname{Tr}\!\left[  \Gamma_{AB}\left(
\sqrt{\rho_{A}}U_{A}^{\dag}\left(  U_{A}\rho_{A}U_{A}^{\dag}\right)
^{\frac{1-\alpha}{\alpha}}U_{A}\sqrt{\rho_{A}}\otimes\sigma_{B}^{\frac
{1-\alpha}{\alpha}}\right)  \right]  \\
&  =\frac{\alpha}{\alpha-1}\log\operatorname{Tr}\!\left[  \Gamma_{AB}\left(
\rho_{A}^{\frac{1}{\alpha}}\otimes\sigma_{B}^{\frac{1-\alpha}{\alpha}}\right)
\right]  \\
&  =\frac{\alpha}{\alpha-1}\log\operatorname{Tr}\!\left[  \Gamma_{AB}\left(
\rho_{A}^{\frac{1}{\alpha}}T(\sigma)_{A}^{\frac{1-\alpha}{\alpha}}\otimes
\1_{B}\right)  \right]  \\
&  =\frac{\alpha}{\alpha-1}\log\operatorname{Tr}\!\left[  \rho_{A}^{\frac
{1}{\alpha}}T(\sigma)_{A}^{\frac{1-\alpha}{\alpha}}\right].
\end{align}
In the above, $T(\sigma)$ denotes the transpose of the state $\sigma$.
Now taking a minimum over every state~$\sigma$ and applying \cite[Lemma~12]{muller2013quantum}, we find that%
\begin{align}
  \min_{\sigma}\frac{\alpha}{\alpha-1}\log\operatorname{Tr}\!\left[
\rho_A^{\frac{1}{\alpha}}T(\sigma)_{A}^{\frac{1-\alpha}{\alpha}}\right]
&  =\frac{\alpha}{\alpha-1}\log\left\Vert \rho_A^{\frac{1}{\alpha}%
}\right\Vert _{\alpha/\left(  2\alpha-1\right)  }\\
&  =\frac{\alpha}{\alpha-1}\frac{2\alpha-1}{\alpha}\log\operatorname{Tr}%
[(\rho_A^{\frac{1}{\alpha}})^{\alpha/\left(  2\alpha-1\right)  }]\\
&  =\frac{2\alpha-1}{\alpha-1}\log\operatorname{Tr}[\rho_A^{\frac{1}%
{2\alpha-1}}]\\
&  =\frac{2}{1-\left(  1/\left[  2\alpha-1\right]  \right)  }\log\operatorname{Tr}[\rho_A^{\frac{1}{2\alpha-1}}]\\
&  =2H_{\left[  2\alpha-1\right]  ^{-1}}(\rho_A).
\end{align}
The conclusion that $H_{\left[  2\alpha-1\right]  ^{-1}}(\rho_A)= H_{\left[  2\alpha-1\right]  ^{-1}}(\psi_A)$ follows because $\rho_A$ and $\psi_A$ have the same spectrum.

\section{Proof of Lemma~\ref{lem:geometric-Renyi-MI-pure}}

\label{app:proof-lemma-geometric-Rny-MI-pure}

Consider that, by similar reasoning employed in
Appendix~\ref{app:proof-lemma-sw-Rny-MI-pure},
\begin{equation}
\psi_{AB}=\psi_{A}^{1/2}U_{A}\Gamma_{AB}U_{A}^{\dag}\psi_{A}^{1/2},
\end{equation}
where $U_{A}$ is a unitary, $\psi_{A}$ is the marginal density operator of
$\psi_{AB}$ on system $A$,
\begin{equation}
\Gamma_{AB}\coloneqq|\Gamma\rangle\!\langle\Gamma|_{AB},\qquad|\Gamma
\rangle_{AB}\coloneqq\sum_{i=1}^{d}|i\rangle_{A}|i\rangle_{B},
\end{equation}
$\left\{  |i\rangle_{A}\right\}  _{i}$ and $\left\{  |i\rangle_{B}\right\}
_{i}$ are the orthonormal bases for the Schmidt decomposition of $\psi_{AB}$,
and $d$ is the Schmidt rank of $\psi_{AB}$. Then, by noting that it suffices
to optimize over every state~$\sigma_{B}$ having the same support as $\psi_{B}$, we
find that%
\begin{align}
\widehat{I}_{\alpha}(A;B)_{\psi} &  \coloneqq\min_{\sigma_{B}}\widehat
{D}_{\alpha}(\psi_{AB}\Vert\psi_{A}\otimes\sigma_{B})\\
&  =\min_{\sigma_{B}}\frac{1}{\alpha-1}\log\operatorname{Tr}\!\left[  \left(
\psi_{A}\otimes\sigma_{B}\right)  \left[  \left(  \psi_{A}\otimes\sigma
_{B}\right)  ^{-1/2}\psi_{AB}\left(  \psi_{A}\otimes\sigma_{B}\right)
^{-1/2}\right]  ^{\alpha}\right]
\end{align}
Consider that%
\begin{align}
&  \operatorname{Tr}\!\left[  \left(  \psi_{A}\otimes\sigma_{B}\right)
\left[  \left(  \psi_{A}^{-1/2}\otimes\sigma_{B}^{-1/2}\right)  \psi
_{AB}\left(  \psi_{A}^{-1/2}\otimes\sigma_{B}^{-1/2}\right)  \right]
^{\alpha}\right]  \nonumber\\
&  =\operatorname{Tr}\!\left[  \left(  \psi_{A}\otimes\sigma_{B}\right)
\left[  \left(  \psi_{A}^{-1/2}\otimes\sigma_{B}^{-1/2}\right)  \psi_{A}%
^{1/2}U_{A}\Gamma_{AB}U_{A}^{\dag}\psi_{A}^{1/2}\left(  \psi_{A}^{-1/2}%
\otimes\sigma_{B}^{-1/2}\right)  \right]  ^{\alpha}\right]  \\
&  =\operatorname{Tr}\!\left[  \left(  \psi_{A}\otimes\sigma_{B}\right)
\left[  \left(  U_{A}\otimes\sigma_{B}^{-1/2}\right)  \Gamma_{AB}\left(
U_{A}^{\dag}\otimes\sigma_{B}^{-1/2}\right)  \right]  ^{\alpha}\right]  \\
&  =\operatorname{Tr}\!\left[  \left(  \psi_{A}\otimes\sigma_{B}\right)
U_{A}\left[  \sigma_{B}^{-1/2}\Gamma_{AB}\sigma_{B}^{-1/2}\right]  ^{\alpha
}U_{A}^{\dag}\right]  \\
&  =\operatorname{Tr}\!\left[  \left(  U_{A}^{\dag}\psi_{A}U_{A}\otimes
\sigma_{B}\right)  \left[  \frac{\sigma_{B}^{-1/2}\Gamma_{AB}\sigma_{B}%
^{-1/2}}{\operatorname{Tr}[\sigma_{B}^{-1}]}\right]  ^{\alpha}\right]  \left(
\operatorname{Tr}[\sigma_{B}^{-1}]\right)  ^{\alpha}\\
&  =\operatorname{Tr}\!\left[  \left(  U_{A}^{\dag}\psi_{A}U_{A}\otimes
\sigma_{B}\right)  \frac{\sigma_{B}^{-1/2}\Gamma_{AB}\sigma_{B}^{-1/2}%
}{\operatorname{Tr}[\sigma_{B}^{-1}]}\right]  \left(  \operatorname{Tr}%
[\sigma_{B}^{-1}]\right)  ^{\alpha}\\
&  =\operatorname{Tr}\!\left[  \left(  U_{A}^{\dag}\psi_{A}U_{A}\otimes
\sigma_{B}\right)  \sigma_{B}^{-1/2}\Gamma_{AB}\sigma_{B}^{-1/2}\right]
\left(  \operatorname{Tr}[\sigma_{B}^{-1}]\right)  ^{\alpha-1}\\
&  =\operatorname{Tr}\!\left[  U_{A}^{\dag}\psi_{A}U_{A}\Gamma_{AB}\right]
\left(  \operatorname{Tr}[\sigma_{B}^{-1}]\right)  ^{\alpha-1}\\
&  =\left(  \operatorname{Tr}[\sigma_{B}^{-1}]\right)  ^{\alpha-1}.
\end{align}
In the fourth step, we are using the fact that $\sigma_{B}^{-1/2}%
|\Gamma\rangle_{AB}$ is a vector with norm%
\begin{align}
\left\Vert \sigma_{B}^{-1/2}|\Gamma\rangle_{AB}\right\Vert _{2} &
=\sqrt{\langle\Gamma|_{AB}\sigma_{B}^{-1/2}\sigma_{B}^{-1/2}|\Gamma
\rangle_{AB}}\\
&  =\sqrt{\langle\Gamma|_{AB}\sigma_{B}^{-1}|\Gamma\rangle_{AB}}\\
&  =\sqrt{\operatorname{Tr}[\sigma_{B}^{-1}]}.
\end{align}
Then it follows that%
\begin{align}
\widehat{I}_{\alpha}(A;B)_{\psi} &  =\min_{\sigma_{B}}\frac{1}{\alpha-1}%
\log\left(  \operatorname{Tr}[\sigma_{B}^{-1}]\right)  ^{\alpha-1}\\
&  =\min_{\sigma_{B}}\log\operatorname{Tr}[\sigma_{B}^{-1}].
\end{align}
The minimum value occurs when $\sigma_{B}=\1_{B}/d_{B}$ (the maximally mixed
state in system $B$, where we again note that $B$ is restricted to the support of $\psi_B$ so that $d_B$ is the dimension of the support of~$\psi_B$). To conclude this final step, one can use the Lagrange
multiplier method. So then the equality in the statement of
Lemma~\ref{lem:geometric-Renyi-MI-pure} follows.

\section{Proof of Proposition~\ref{prop:subadditivity}}
\label{appx:subadditivity}

Let $\sigma_{A_1B_1B_1'}$ and $\sigma_{A_2B_2B_2'}$ be arbitrary extensions of $\rho_{A_1B_1}$ and
$\rho_{A_2B_2}$, respectively. Consider the state
$\sigma_{A_1B_1B_1'}\ox\sigma_{A_2B_2B_2'}$. Then it follows that
\begin{align}
    \tr_{B_1'B_2'}[\sigma_{A_1B_1B_1'}\ox\sigma_{A_2B_2B_2'}]
    =  \tr_{B_1'}[\sigma_{A_1B_1B_1'}]\ox\tr_{B_2'}[\sigma_{A_2B_2B_2'}]
    = \rho_{A_1B_1}\ox\rho_{A_2B_2},
\end{align}
so that $\sigma_{A_1B_1B_1'}\ox\sigma_{A_2B_2B_2'}$ is an extension of
$\rho_{A_1B_1}\ox\rho_{A_2B_2}$. Thus
\begin{align}
    E_{\a}^u(\rho_{A_1B_1}\ox\rho_{A_2B_2})
&\leq \frac{1}{2} D_{\a}(\rho_{A_1B_1}\ox\rho_{A_2B_2}\|\sigma_{A_1B_1'}\ox\sigma_{A_2B_2'}) \\
&=    \frac{1}{2} D_{\a}(\rho_{A_1B_1}\|\sigma_{A_1B_1'})
    + \frac{1}{2} D_{\a}(\rho_{A_2B_2}\|\sigma_{A_2B_2'}).
\end{align}
Taking an infimum over all extensions $\sigma_{A_1B_1B_1'}$ and $\sigma_{A_2B_2B_2'}$ leads to the
inequality in \eqref{eq:petz-subadd}.
The inequalities in \eqref{eq:sandwiched-subadd} and \eqref{eq:geometric-subadd} are established similarly.

\section{Derivations of dual SDPs and strong duality for max- and min-unextendible entanglement}

\label{app:dual-SDP-derivations}

We first derive the dual SDP\ in \eqref{eq:max-dual}, by means of the Lagrange multiplier
method. Consider that%
\begin{align}
&  \sup_{\lambda,\sigma_{ABB^{\prime}}\geq0}\left\{  \lambda:\lambda\rho
_{AB}\leq\operatorname{Tr}_{B}[\sigma_{ABB^{\prime}}],\ \rho_{AB}%
=\operatorname{Tr}_{B^{\prime}}[\sigma_{ABB^{\prime}}]\right\}  \nonumber\\
&  =\sup_{\lambda,\sigma_{ABB^{\prime}}\geq0}\inf_{\substack{X_{AB}%
\geq0,\\Y_{AB}\in\text{Herm}}}\left\{
\begin{array}
[c]{c}%
\lambda+\operatorname{Tr}[X_{AB}(\operatorname{Tr}_{B}[\sigma_{ABB^{\prime}%
}]-\lambda\rho_{AB})]\\
+\operatorname{Tr}[Y_{AB}(\rho_{AB}-\operatorname{Tr}_{B^{\prime}}%
[\sigma_{ABB^{\prime}}])]
\end{array}
\right\}  \\
&  =\sup_{\lambda,\sigma_{ABB^{\prime}}\geq0}\inf_{\substack{X_{AB}%
\geq0,\\Y_{AB}\in\text{Herm}}}\left\{
\begin{array}
[c]{c}%
\operatorname{Tr}[Y_{AB}\rho_{AB}]+\lambda\left(  1-\operatorname{Tr}%
[X_{AB}\rho_{AB}]\right)  \\
+\operatorname{Tr}[\left(  X_{AB^{\prime}}\otimes \1_{B}-Y_{AB}\otimes
\1_{B^{\prime}}\right)  \sigma_{ABB^{\prime}}]
\end{array}
\right\}  \\
&  \leq\inf_{\substack{X_{AB}\geq0,\\Y_{AB}\in\text{Herm}}}\sup_{\lambda
,\sigma_{ABB^{\prime}}\geq0}\left\{
\begin{array}
[c]{c}%
\operatorname{Tr}[Y_{AB}\rho_{AB}]+\lambda\left(  1-\operatorname{Tr}%
[X_{AB}\rho_{AB}]\right)  \\
+\operatorname{Tr}[\left(  X_{AB^{\prime}}\otimes \1_{B}-Y_{AB}\otimes
\1_{B^{\prime}}\right)  \sigma_{ABB^{\prime}}]
\end{array}
\right\}  \\
&  =\inf_{\substack{X_{AB}\geq0,\\Y_{AB}\in\text{Herm}}}\left\{
\begin{array}
[c]{c}%
\operatorname{Tr}[Y_{AB}\rho_{AB}]:\operatorname{Tr}[X_{AB}\rho_{AB}]\geq1,\\
X_{AB^{\prime}}\otimes \1_{B}\leq Y_{AB}\otimes \1_{B^{\prime}}%
\end{array}
\right\}  \\
&  =\inf_{X_{AB},Y_{AB}\geq0}\left\{
\begin{array}
[c]{c}%
\operatorname{Tr}[Y_{AB}\rho_{AB}]:\operatorname{Tr}[X_{AB}\rho_{AB}]\geq1,\\
X_{AB^{\prime}}\otimes \1_{B}\leq Y_{AB}\otimes \1_{B^{\prime}}%
\end{array}
\right\}  .
\end{align}
The first equality follows by introducing the Lagrange multipliers $X_{AB}%
\geq0$ and $Y_{AB}\in$Herm. Indeed, the constraint $\lambda\rho_{AB}%
\leq\operatorname{Tr}_{B}[\sigma_{ABB^{\prime}}]$ does not hold if and only if
\begin{equation}
\inf_{X_{AB}\geq0}\operatorname{Tr}[X_{AB}(\operatorname{Tr}_{B}%
[\sigma_{ABB^{\prime}}]-\lambda\rho_{AB})]=-\infty,
\end{equation}
and the constraint
$\rho_{AB}=\operatorname{Tr}_{B^{\prime}}[\sigma_{ABB^{\prime}}]$ does not
hold if and only if
\begin{equation}
\inf_{Y_{AB}\in\text{Herm}}\operatorname{Tr}[Y_{AB}%
(\rho_{AB}-\operatorname{Tr}_{B^{\prime}}[\sigma_{ABB^{\prime}}])]=  -\infty.
\end{equation}
The
second equality follows from basic algebra. The inequality follows from the
max-min inequality. The third equality follows by intepreting $\lambda
,\sigma_{ABB^{\prime}}\geq0$ as Lagrange multipliers. Indeed, the constraint
$\operatorname{Tr}[X_{AB}\rho_{AB}]\geq1$ does not hold if and only if
\begin{equation}
\sup_{\lambda\geq0}\lambda\left(  1-\operatorname{Tr}[X_{AB}\rho
_{AB}]\right)  =+\infty,
\end{equation}
and the constraint $X_{AB^{\prime}}\otimes \1_{B}\leq
Y_{AB}\otimes \1_{B^{\prime}}$ does not hold if and only if
\begin{equation}
\sup
_{\sigma_{ABB^{\prime}}\geq0}\operatorname{Tr}[\left(  X_{AB^{\prime}}\otimes
\1_{B}-Y_{AB}\otimes \1_{B^{\prime}}\right)  \sigma_{ABB^{\prime}}]=+\infty.
\end{equation}
The final equality follows because the constraint $X_{AB^{\prime}}\otimes
\1_{B}\leq Y_{AB}\otimes \1_{B^{\prime}}$ implies that $Y_{AB}\geq0$ since
$X_{AB^{\prime}}\otimes \1_{B}\geq0$ and the partial trace is a completely
positive map.

We can actually conclude that strong duality holds (i.e., equality in all
steps above) because Slater's condition holds. Indeed, a feasible choice for
the primal SDP\ is to pick $\sigma_{ABB^{\prime}}=\rho_{AB}\otimes
\pi_{B^{\prime}}$, where $\pi_{B^{\prime}}$ is the maximally mixed state, and
$\lambda=\frac{1}{d_{B}^{2}}$. We are then guaranteed that $\lambda\rho
_{AB}\leq\operatorname{Tr}_{B}[\sigma_{ABB^{\prime}}]$ for these choices
because $\frac{1}{d_{B}^{2}}\rho_{AB}\leq\frac{1}{d_{B}^{2}}\sum_{i=1}%
^{d_{B}^{2}}U_{B}^{i}\rho_{AB}U_{B}^{i\dag}=\rho_{A}\otimes\pi_{B^{\prime}}$,
where $\left\{  U_{B}^{i}\right\}  _{i=1}^{d_{B}^{2}}$ is the set of
Heisenberg--Weyl unitaries. A strictly feasible choice for the dual SDP\ is
$X_{AB}=2\1_{AB}$ and $Y_{AB}=3\1_{AB}$.

Let us now derive the dual SDP\ in \eqref{eq:min-dual}, again by means of the Lagrange
multiplier method. Consider that%
\begin{align}
&  \sup_{\sigma_{ABB^{\prime}}\geq0}\left\{  \operatorname{Tr}[\Pi
_{AB^{\prime}}^{\rho}\operatorname{Tr}_{B}[\sigma_{ABB^{\prime}}%
]]:\operatorname{Tr}_{B^{\prime}}[\sigma_{ABB^{\prime}}]=\rho_{AB}\right\}
\nonumber\\
&  =\sup_{\sigma_{ABB^{\prime}}\geq0}\inf_{X_{AB}\in\text{Herm}}\left\{
\operatorname{Tr}[\Pi_{AB^{\prime}}^{\rho}\operatorname{Tr}_{B}[\sigma
_{ABB^{\prime}}]]+\operatorname{Tr}[X_{AB}(\rho_{AB}-\operatorname{Tr}%
_{B^{\prime}}[\sigma_{ABB^{\prime}}])]\right\}  \\
&  =\sup_{\sigma_{ABB^{\prime}}\geq0}\inf_{X_{AB}\in\text{Herm}}\left\{
\operatorname{Tr}[X_{AB}\rho_{AB}]+\operatorname{Tr}[\left(  \Pi_{AB^{\prime}%
}^{\rho}\otimes \1_{B}-X_{AB}\otimes \1_{B^{\prime}}\right)  \sigma
_{ABB^{\prime}}]\right\}  \\
&  \leq\inf_{X_{AB}\in\text{Herm}}\sup_{\sigma_{ABB^{\prime}}\geq0}\left\{
\operatorname{Tr}[X_{AB}\rho_{AB}]+\operatorname{Tr}[\left(  \Pi_{AB^{\prime}%
}^{\rho}\otimes \1_{B}-X_{AB}\otimes \1_{B^{\prime}}\right)  \sigma
_{ABB^{\prime}}]\right\}  \\
&  =\inf_{X_{AB}\in\text{Herm}}\left\{  \operatorname{Tr}[X_{AB}\rho_{AB}%
]:\Pi_{AB^{\prime}}^{\rho}\otimes \1_{B}\leq X_{AB}\otimes \1_{B^{\prime}%
}\right\}  \\
&  =\inf_{X_{AB}\geq0}\left\{  \operatorname{Tr}[X_{AB}\rho_{AB}%
]:\Pi_{AB^{\prime}}^{\rho}\otimes \1_{B}\leq X_{AB}\otimes \1_{B^{\prime}%
}\right\}  .
\end{align}
The first equality follows because the constraint $\operatorname{Tr}%
_{B^{\prime}}[\sigma_{ABB^{\prime}}]=\rho_{AB}$ does not hold if and only if
\begin{equation}
\inf_{X_{AB}\in\text{Herm}}\operatorname{Tr}[X_{AB}(\rho_{AB}%
-\operatorname{Tr}_{B^{\prime}}[\sigma_{ABB^{\prime}}])]=-\infty.
\end{equation}
The second
equality follows from basic algebra. The inequality follows from the max-min
inequality. The third equality follows because the constraint $\Pi
_{AB^{\prime}}^{\rho}\otimes \1_{B}\leq X_{AB}\otimes \1_{B^{\prime}}$ does not
hold if and only if
\begin{equation}
\sup_{\sigma_{ABB^{\prime}}\geq0}\operatorname{Tr}%
[\left(  \Pi_{AB^{\prime}}^{\rho}\otimes \1_{B}-X_{AB}\otimes \1_{B^{\prime}%
}\right)  \sigma_{ABB^{\prime}}]=+\infty.
\end{equation}
The final equality follows because
the constraint $\Pi_{AB^{\prime}}^{\rho}\otimes \1_{B}\leq X_{AB}\otimes
\1_{B^{\prime}}$ implies that $X_{AB}\geq0$ since $\Pi_{AB^{\prime}}^{\rho
}\otimes \1_{B}\geq0$ and the partial trace is a completely positive map.

Strong duality holds (i.e., equality in all steps above) because Slater's
condition holds. Indeed, a feasible choice for the primal SDP\ is
$\sigma_{ABB^{\prime}}=\rho_{AB}\otimes\pi_{B^{\prime}}$. A strictly feasible
choice for the dual SDP\ is $X_{AB}=2\1_{AB}$.

\section{Proof of Proposition~\ref{prop:max-dditivity}}
\label{appx:max-dditivity}

Subadditivity follows from investigating the primal SDP.
Suppose that $(\lambda_1,\sigma_{A_1B_1B_1^\prime})$ and $(\lambda_2,\sigma_{A_2B_2B_2^\prime})$
achieve $E_{\max}^u\left(\rho_{A_1B_1}\right)$ and $E_{\max}^u\left(\rho_{A_2B_2}\right)$, respectively.
As
\begin{align}
  \lambda_1\lambda_2\rho_{A_1B_1}\ox\rho_{A_2B_2}
= \lambda_1\rho_{A_1B_1}\ox\lambda_2\rho_{A_2B_2}
\leq \sigma_{A_1B_1'}\ox\sigma_{A_2B_2'},
\end{align}
one can check that $(\lambda_1\lambda_2,
\sigma_{A_1B_1B_1^\prime}\ox\sigma_{A_2B_2B_2^\prime})$ forms a feasible solution. Thus
\begin{align*}
      2^{-2E_{\max}^u\left(\rho_{A_1B_1}\ox\rho_{A_2B_2}\right)}
\geq \lambda_1\lambda_2
= 2^{-2E_{\max}^u\left(\rho_{A_1B_1}\right) - 2E_{\max}^u\left(\rho_{A_2B_2}\right)},
\end{align*}
which gives $E_{\max}^u\left(\rho_{A_1B_1}\ox\rho_{A_2B_2}\right) \leq
E_{\max}^u\left(\rho_{A_1B_1}\right) + E_{\max}^u\left(\rho_{A_2B_2}\right)$.

The superadditivity is shown by investigating the dual SDP.
Suppose that $(X_{A_1B_1},Y_{A_1B_1})$ and $(X_{A_2B_2},Y_{A_2B_2})$
achieve $E_{\max}^u\left(\rho_{A_1B_1}\right)$ and $E_{\max}^u\left(\rho_{A_2B_2}\right)$, respectively.
As
\begin{align}
\tr\left[\left(\rho_{A_1B_1}\ox\rho_{A_2B_2}\right)
        \left(X_{A_1B_1}\ox X_{A_2B_2}\right)\right]
=  \tr\left[\rho_{A_1B_1}X_{A_1B_1}\right]\tr\left[\rho_{A_2B_2}X_{A_2B_2}\right]  \geq 1
\end{align}
and $X_{A_1B_1'}\ox X_{A_2B_2'} \ox \1_{B_1B_2} \leq Y_{A_1B_1}\ox Y_{A_2B_2} \ox \1_{B_1'B_2'}$, the pair
$(X_{A_1B_1}\ox X_{A_2B_2}, Y_{A_1B_1}\ox Y_{A_2B_2})$ forms a dual feasible solution. Thus
\begin{align}
      2^{-2E_{\max}^u\left(\rho_{A_1B_1}\ox\rho_{A_2B_2}\right)}
\leq  \tr\left[\rho_{A_1B_1}Y_{A_1B_1}\right]\tr\left[\rho_{A_2B_2}Y_{A_2B_2}\right]
= 2^{-2E_{\max}^u\left(\rho_{A_1B_1}\right) - 2E_{\max}^u\left(\rho_{A_2B_2}\right)},
\end{align}
which gives $E_{\max}^u\left(\rho_{A_1B_1}\ox\rho_{A_2B_2}\right) \geq
E_{\max}^u\left(\rho_{A_1B_1}\right) + E_{\max}^u\left(\rho_{A_2B_2}\right)$.

\section{Proof of Proposition~\ref{prop:e-u-min-pure}}
\label{appx:e-u-min-pure}

This proposition follows from~\eqref{eq:e-alpha-unext-pure-states} by taking the limit $\alpha \to 0$. 
Alternatively, we briefly outline another proof as follows. For a pure state $\psi_{AB}$, 
an arbitrary extension of it has the form
$\sigma_{ABB'}\coloneqq \proj{\psi}_{AB}\ox\sigma_{B'}$, where $\sigma_{B'}$ is a state.
Then $\tr_{B}[\sigma_{ABB'}]=\psi_A\ox\sigma_{B'}$.
Suppose that a spectral decomposition of $\sigma_{B'}$ is $\sigma_{B'}=\sum_n p_n\proj{e_n}$. Then
\begin{align}
&\quad \frac{1}{2} D_{\min}(\psi_{AB}\|\psi_A\ox\sigma_{B'}) \notag\\
&=  -\frac{1}{2}\log\sum_{i,j,m,n}\sqrt{\alpha_i\alpha_j}\alpha_m p_n
    \bra{\psi_i}_A \bra{\psi_i}_B
    \left(\proj{\psi_m}_A \ox \proj{e_n}_B\right)
    \ket{\psi_j}_A\ket{\psi_j}_B  \\
&= -\frac{1}{2}\log\sum_{i,n}\alpha_i^2 p_n \vert\langle\psi_i\vert e_n\rangle\vert^2 \\
&\geq -\frac{1}{2}\log\alpha_1^2\sum_{i,n}p_n\vert\langle\psi_i\vert e_n\rangle\vert^2 \\
&= - \log\alpha_1.
\end{align}
The inequality is achieved by choosing $\sigma_{B'} = \proj{\psi_1}_B$. Thus $E_{\min}^u(\psi_{AB}) = - \log
\alpha_1$.

\section{Dual SDP of unextendible fidelity}
\label{app:dual-sdp-unextendible-fidelity}

Here we derive the dual SDP of the primal SDP in~\eqref{eq:fid-primal} for unextendible fidelity. We follow an
 argument similar to that given in~\cite{berta2016}. We first bring the primal program into standard form, which expresses the
primal problem as a maximization of the functional $\tr[AX]$ over $X \geq 0$, subject to the constraint $\Phi(X) = B$. Hence, we set
\begin{align}
    X = \begin{pmatrix}
        X_{11} & Z_{AB} & \cdot \\
        Z_{AB}^\dagger & X_{22} & \cdot \\
        \cdot & \cdot &  \sigma_{ABB'}
    \end{pmatrix},\quad
    A = \frac{1}{2}\begin{pmatrix}
        0 & \1_{AB} & 0 \\
        \1_{AB} & 0 & 0 \\
        0 & 0 & 0
    \end{pmatrix},\quad
    B = \begin{pmatrix}
        \rho_{AB} & 0 & 0 \\
        0 &  0 & 0 \\
        0 & 0 & \rho_{AB}
    \end{pmatrix},
\end{align}
and
\begin{align}\label{eq:phi-x}
    \Phi(X) \coloneqq
\begin{pmatrix}
  X_{11} & 0 & 0 \\
  0 & X_{22} - \tr_{B}[\sigma_{ABB'}] & 0 \\
  0 & 0 & \tr_{B'}[\sigma_{ABB'}]
\end{pmatrix} .
\end{align}
The variables with the placeholder `$\cdot$' are of no interest. The dual SDP is a minimization over self-adjoint
$Y$ of the functional $\tr[BY]$ subject to $\Phi^\dagger(Y) \geq A$. The dual variables and adjoint map can be
determined to be
\begin{align}
Y =
\begin{pmatrix}
    Y_{11} & \cdot & \cdot \\
    \cdot & Y_{22} & \cdot \\
    \cdot & \cdot & Y_{33}
\end{pmatrix}
\text{~~and~~}
\Phi^\dagger(Y) \coloneqq
\begin{pmatrix}
  Y_{11} & 0 & 0 \\
  0 & Y_{22} & 0 \\
  0 & 0 & Y_{33}\ox\1_{B'} - Y_{22}\ox\1_{B}
\end{pmatrix}.
\end{align}
This leads to the following dual problem:
\begin{equation}
\begin{split}
\text{minimize}   &\quad \tr\left[(Y_{11} + Y_{33})\rho_{AB}\right]  \\
\text{subject to} &\quad  \begin{pmatrix} Y_{11} & 0 \\ 0 & Y_{22}\end{pmatrix}
                    \geq \frac{1}{2}\begin{pmatrix} 0 & \1_{AB} \\ \1_{AB} & 0 \end{pmatrix} \\
            &\quad Y_{33}\ox\1_{B'} \geq Y_{22}\ox\1_{B} \\
            &\quad Y_{11}, Y_{22}, Y_{33} \in \operatorname{Herm}(AB).
\end{split}\end{equation}
The Slater condition for strong duality is satisfied, using the fact that the primal problem is feasible and the dual
problem is strictly feasible. To see this, let $\sigma_{B'}$ be a quantum state. Then the operator
\begin{align}
    \tilde{X} = \begin{pmatrix}
        \rho_{AB} & 0 & 0 \\
        0 & \rho_A\ox\sigma_{B'} & 0 \\
        0 & 0 &  \rho_{AB}\ox\sigma_{B'}
    \end{pmatrix}
\end{align}
is primal feasible since $\tilde{X}\geq 0$ and $\Phi(\tilde{X}) = B$. For the dual problem, the operator
\begin{align}
    \tilde{Y} = \begin{pmatrix}
        \1_{AB} & 0 & 0 \\
        0 & \1_{AB} & 0 \\
        0 & 0 &  2\1_{AB}
    \end{pmatrix},
\end{align}
is strictly feasible since
\begin{align}
    \Phi^\dagger(\tilde{Y}) =
\begin{pmatrix}
        \1_{AB} & 0 & 0 \\
        0 & \1_{AB} & 0 \\
        0 & 0 &  \1_{ABB'}
    \end{pmatrix} >
    \frac{1}{2}\begin{pmatrix}
        0 & \1_{AB} & 0 \\
        \1_{AB} & 0 & 0 \\
        0 & 0 & 0
    \end{pmatrix} \equiv A.
\end{align}
By performing the substitutions $Y_{11} \to \frac{1}{2} Y_{11}$ and $Y_{33} \to \frac{1}{2} Y_{33}$ and noting that
\begin{equation}
\begin{pmatrix} Y_{11} & 0 \\ 0 & Y_{22}\end{pmatrix}
                    \geq \begin{pmatrix} 0 & \1_{AB} \\ \1_{AB} & 0 \end{pmatrix} \quad \Leftrightarrow \quad 
                    \begin{pmatrix} Y_{11} & -\1_{AB} \\ -\1_{AB} & Y_{22}\end{pmatrix}
                    \geq 0,
\end{equation}
we finally obtain the dual stated in \eqref{eq:fid-dual}:
\begin{equation}\label{eq:dual1}
\begin{split}
\text{minimize}   &\quad \frac{1}{2}\tr\left[(Y_{11} + Y_{33})\rho_{AB}\right]  \\
\text{subject to} &\quad  \begin{pmatrix} Y_{11} & -\1_{AB} \\ -\1_{AB} & Y_{22}\end{pmatrix}
                    \geq 0 \\
            &\quad Y_{33}\ox\1_{B'} \geq Y_{22}\ox\1_{B} \\
            &\quad Y_{11}, Y_{22}, Y_{33} \geq 0.
\end{split}\end{equation}
In this final step, we also used the facts that
\begin{align}
\begin{pmatrix} Y_{11} & -\1_{AB} \\ -\1_{AB} & Y_{22}\end{pmatrix}
                    \geq 0 \quad & \Rightarrow \quad 
                    Y_{11} , Y_{22} \geq 0,\\
                    Y_{22} \geq 0, \quad  Y_{33}\ox\1_{B'} \geq Y_{22}\ox\1_{B} \quad & \Rightarrow \quad Y_{33}\geq 0.
\end{align}

The dual problem in~\eqref{eq:dual1} can be further simplified since the first matrix inequality therein holds if and only if
$Y_{11},Y_{22}> 0$ and $Y_{22} \geq Y_{11}^{-1}$~\cite{watrous2013}. Without loss of generality, we can choose $Y_{22} =
Y_ {11}^{-1}$, and the problem simplifies to
\begin{equation}\label{eq:dual3}
\begin{split}
\text{infimum}   &\quad \frac{1}{2}\tr\left[Y_{AB'}^{-1}\rho_{AB'}\right] + \frac{1}{2}\tr\left[Z_{AB}\rho_{AB}\right] \\
\text{subject to} &\quad Z_{AB}\ox\1_{B'} \geq Y_{AB'}\ox\1_{B} \\
            &\quad Y_{AB'}, Z_{AB} > 0.
\end{split}\end{equation}
By the
arithmetic-geometric mean inequality, it holds that
\begin{align}
   \frac{1}{2}\tr\left[Y_{AB'}^{-1}\rho_{AB'}\right] + \frac{1}{2}\tr\left[Z_{AB}\rho_{AB}\right]
\geq \sqrt{\tr[Y_{AB'}^{-1}\rho_{AB'}]\tr[Z_{AB}\rho_{AB}]},
\end{align}
for every $Y_{AB'} > 0$, with equality when the two terms are equal. For every feasible pair $(Y_{AB'}, Z_{AB})$, there
exists a constant $\lambda>0$ such that the two trace terms evaluated on $(\lambda Y_{AB'},\lambda Z_{AB})$ are equal.
Hence, we can restrict the optimization to such rescaled pairs of operators, resulting
\begin{equation}\label{eq:dual4}
\begin{split}
\text{infimum}   &\quad \sqrt{\tr[Y_{AB'}^{-1}\rho_{AB'}]\tr[Z_{AB}\rho_{AB}]}  \\
\text{subject to} &\quad Z_{AB}\ox\1_{B'} \geq Y_{AB'}\ox\1_{B} \\
            &\quad Y_{AB'}, Z_{AB} > 0.
\end{split}
\end{equation}
We remark that this equivalent representation of the dual problem is essential when proving the multiplicativity of
unextendible fidelity.

\section{Proof of Proposition~\ref{prop:fid-multiplicativity}}
\label{appx:fid-multiplicativity}
For the ``$\geq$'' part, suppose that $\rho_{A_1B_1B_1'}$ and $\rho_{A_2B_2B_2'}$ achieve $F^u\left(\rho_{A_1B_1}\right)$ and
$F^u\left(\rho_{A_2B_2}\right)$, respectively. Since
\begin{align}
  \tr_{B_1'B_2'}\!\left[\rho_{A_1B_1B_1'}\ox\rho_{A_2B_2B_2'}\right] = \rho_{A_1B_1}\ox\rho_{A_2B_2},
\end{align}
the state $\rho_{A_1B_1B_1'}\ox\rho_{A_2B_2B_2'}$ is a feasible solution for $F^u(\rho_{A_1B_1}\ox\rho_{A_2B_2})$.
So
\begin{align}
  F^u(\rho_{A_1B_1}\ox\rho_{A_2B_2})
& \geq F(\rho_{A_1B_1}\ox\rho_{A_2B_2}, \rho_{A_1B_1'}\ox\rho_{A_2B_2'}) \\
& = F(\rho_{A_1B_1}, \rho_{A_1B_1'})F(\rho_{A_2B_2}, \rho_{A_2B_2'}) \\
& = F^u(\rho_{A_1B_1})F^u(\rho_{A_2B_2}),
\end{align}
where the first equality follows from multiplicativity of fidelity.

For the ``$\leq$'' part, we employ the dual representation in \eqref{eq:dual-rep-unext-fid}.
Let $(Y_{A_1B_1'}, Z_{A_1B_1})$ and $(Y_{A_2B_2'}, Z_{A_2B_2})$ be feasible solutions for
$F^u\left(\rho_{A_1B_1}\right)$ and $F^u\left(\rho_{A_2B_2}\right)$, respectively. Set $Y_{A_1A_2B_1'B_2'} \equiv
Y_{A_1B_1'}\ox Y_{A_2B_2'}$ and $Z_{A_1A_2B_1B_2} \equiv Z_{A_1B_1}\ox Z_{A_2B_2}$. Since
\begin{align}
  Z_{A_1A_2B_1B_2} \ox \1_{B_1'B_2'}
& = (Z_{A_1B_1}\ox\1_{B_1'})\ox(Z_{A_2B_2}\ox\1_{B_2'}) \\
& \geq (Y_{A_1B_1'}\ox\1_{B_1})\ox(Y_{A_2B_2'}\ox\1_{B_2})
 = Y_{A_1A_2B_1'B_2'},
\end{align}
it follows that
$(Y_{A_1A_2B_1'B_2'}, Z_{A_1A_2B_1B_2})$ is a feasible solution. Then
\begin{align}
     F^u(\rho_{A_1B_1}\ox\rho_{A_2B_2})
& \leq \sqrt{\tr[Y_{A_1A_2B_1'B_2'}^{-1}(\rho_{A_1B_1}\ox\rho_{A_2B_2})]
        \tr[Z_{A_1A_2B_1B_2}(\rho_{A_1B_1}\ox\rho_{A_2B_2})]} \\
& =    \sqrt{\tr[Y_{A_1B_1'}^{-1}\rho_{A_1B_1}]\tr[Z_{A_1B_1}\rho_{A_1B_1}]}
        \sqrt{\tr[Y_{A_2B_2'}^{-1}\rho_{A_2B_2}]\tr[Z_{A_2B_2}\rho_{A_2B_2}]}.
\end{align}
Since the above inequality holds for all feasible solutions, we find that
\begin{align}
    F^u(\rho_{A_1B_1}\ox\rho_{A_2B_2}) \leq F^u(\rho_{A_1B_1})F^u(\rho_{A_2B_2}).
\end{align}
This concludes the proof.

\section{Proof of Proposition~\ref{prop:lower-bound}}
\label{appx:lower-bound}

Since a maximally entangled state $\Phi_{AB}\equiv \Phi^K_{AB}$
necessarily has an extension of the form $\Phi_{AB}\otimes\sigma_E$
for some quantum state $\sigma_E$,
the only possible extensions of $\gamma_{ABA^{\prime}B^{\prime}}$ have the following form:%
\begin{align}
\gamma_{ABEA^{\prime}B^{\prime}E^{\prime}}  & =U_{ABA^{\prime}B^{\prime}%
}\left(  \Phi_{AB}\otimes\sigma_{EA^{\prime}B^{\prime}E^{\prime}}\right)
U_{ABA^{\prime}B^{\prime}}^{\dag}\\
& =\frac{1}{d}\sum_{i,j}|i\rangle\!\langle j|_{A}\otimes|i\rangle\!\langle
j|_{B}\otimes U_{A^{\prime}B^{\prime}}^{ii}\sigma_{EA^{\prime}B^{\prime
}E^{\prime}}(U_{A^{\prime}B^{\prime}}^{jj})^{\dag},
\end{align}
where $\sigma_{EA^{\prime}B^{\prime}E^{\prime}}$ is an extension of
$\sigma_{A^{\prime}B^{\prime}}$ such that $E\simeq B$ and $E^{\prime}\simeq
B^{\prime}$. Observe that if we define a different twisting unitary
$V_{ABA^{\prime}B^{\prime}}$ as%
\begin{align}
V_{ABA^{\prime}B^{\prime}}  & \coloneqq \sum_{i,j}|i\rangle\!\langle i|_{A}%
\otimes|j\rangle\!\langle j|_{B}\otimes U_{A^{\prime}B^{\prime}}^{ii}\\
& =\sum_{i}|i\rangle\!\langle i|_{A}\otimes \1_{B}\otimes U_{A^{\prime}B^{\prime
}}^{ii},
\end{align}
then%
\begin{align}
V_{ABA^{\prime}B^{\prime}}\left(  \Phi_{AB}\otimes\sigma_{EA^{\prime}%
B^{\prime}E^{\prime}}\right)  V_{ABA^{\prime}B^{\prime}}^{\dag}  & =\frac
{1}{d}\sum_{i,j}|i\rangle\!\langle j|_{A}\otimes|i\rangle\!\langle j|_{B}\otimes
U_{A^{\prime}B^{\prime}}^{ii}\sigma_{EA^{\prime}B^{\prime}E^{\prime}%
}(U_{A^{\prime}B^{\prime}}^{jj})^{\dag}\\
& =\gamma_{ABEA^{\prime}B^{\prime}E^{\prime}}.%
\end{align}
We have%
\begin{align}
\gamma_{AEA^{\prime}E^{\prime}}  & =\operatorname{Tr}_{BB^{\prime}}%
[\gamma_{ABEA^{\prime}B^{\prime}E^{\prime}}]\\
& =\frac{1}{d}\sum_{i}|i\rangle\!\langle i|_{A}\otimes\operatorname{Tr}%
_{B^{\prime}}[U_{A^{\prime}B^{\prime}}^{ii}\sigma_{EA^{\prime}B^{\prime
}E^{\prime}}(U_{A^{\prime}B^{\prime}}^{ii})^{\dag}].
\end{align}
Now consider that%
\begin{align}
 \mathbf{D}(\gamma_{ABA^{\prime}B^{\prime}}\Vert\gamma_{AEA^{\prime}%
E^{\prime}}) 
& =\mathbf{D}(V_{ABA^{\prime}B^{\prime}}\left(  \Phi_{AB}\otimes
\sigma_{A^{\prime}B^{\prime}}\right)  V_{ABA^{\prime}B^{\prime}}^{\dag}%
\Vert\gamma_{AEA^{\prime}E^{\prime}})\\
& =\mathbf{D}(\Phi_{AB}\otimes\sigma_{A^{\prime}B^{\prime}}\Vert
V_{AEA^{\prime}E^{\prime}}^{\dag}\gamma_{AEA^{\prime}E^{\prime}}%
V_{AEA^{\prime}E^{\prime}})\\
& \geq\mathbf{D}(\Phi_{AB}\Vert\operatorname{Tr}_{A^{\prime}E^{\prime}%
}[V_{AEA^{\prime}E^{\prime}}^{\dag}\gamma_{AEA^{\prime}E^{\prime}%
}V_{AEA^{\prime}E^{\prime}}]).
\end{align}
where the inequality follows from monotonicity under partial trace. Observe
that%
\begin{align}
& V_{AEA^{\prime}E^{\prime}}^{\dag}\gamma_{AEA^{\prime}E^{\prime}%
}V_{AEA^{\prime}E^{\prime}}\nonumber\\
& =\left(  \sum_{i^{\prime}}|i^{\prime}\rangle\!\langle i^{\prime}|_{A}\otimes
\1_{E}\otimes U_{A^{\prime}E^{\prime}}^{i^{\prime}i^{\prime}\dag}\right)
\left(  \frac{1}{d}\sum_{i}|i\rangle\!\langle i|_{A}\otimes\operatorname{Tr}%
_{B^{\prime}}[U_{A^{\prime}B^{\prime}}^{ii}\sigma_{EA^{\prime}B^{\prime
}E^{\prime}}(U_{A^{\prime}B^{\prime}}^{ii})^{\dag}]\right)  \nonumber\\
& \qquad\times\left(  \sum_{i^{\prime\prime}}|i^{\prime\prime}\rangle\!\langle
i^{\prime\prime}|_{A}\otimes \1_{E}\otimes U_{A^{\prime}E^{\prime}}%
^{i^{\prime\prime}i^{\prime\prime}}\right)  \\
& =\frac{1}{d}\sum_{i}|i\rangle\!\langle i|_{A}\otimes U_{A^{\prime}E^{\prime}%
}^{ii\dag}\operatorname{Tr}_{B^{\prime}}[U_{A^{\prime}B^{\prime}}^{ii}%
\sigma_{EA^{\prime}B^{\prime}E^{\prime}}(U_{A^{\prime}B^{\prime}}^{ii})^{\dag
}]U_{A^{\prime}E^{\prime}}^{ii},
\end{align}
so that%
\begin{align}
& \operatorname{Tr}_{A^{\prime}E^{\prime}}[V_{AEA^{\prime}E^{\prime}}^{\dag
}\gamma_{AEA^{\prime}E^{\prime}}V_{AEA^{\prime}E^{\prime}}]
\notag \\
& =\operatorname{Tr}_{A^{\prime}E^{\prime}}\!\left[  \frac{1}{d}\sum
_{i}|i\rangle\!\langle i|_{A}\otimes U_{A^{\prime}E^{\prime}}^{ii\dag
}\operatorname{Tr}_{B^{\prime}}[U_{A^{\prime}B^{\prime}}^{ii}\sigma
_{EA^{\prime}B^{\prime}E^{\prime}}(U_{A^{\prime}B^{\prime}}^{ii})^{\dag
}]U_{A^{\prime}E^{\prime}}^{ii}\right]  \\
& =\frac{1}{d}\sum_{i}|i\rangle\!\langle i|_{A}\otimes\operatorname{Tr}%
_{A^{\prime}E^{\prime}}\!\left[  U_{A^{\prime}E^{\prime}}^{ii\dag}%
\operatorname{Tr}_{B^{\prime}}[U_{A^{\prime}B^{\prime}}^{ii}\sigma
_{EA^{\prime}B^{\prime}E^{\prime}}(U_{A^{\prime}B^{\prime}}^{ii})^{\dag
}]U_{A^{\prime}E^{\prime}}^{ii}\right]  \\
& =\frac{1}{d}\sum_{i}|i\rangle\!\langle i|_{A}\otimes\operatorname{Tr}%
_{A^{\prime}E^{\prime}}\!\left[  \operatorname{Tr}_{B^{\prime}}[U_{A^{\prime
}B^{\prime}}^{ii}\sigma_{EA^{\prime}B^{\prime}E^{\prime}}(U_{A^{\prime
}B^{\prime}}^{ii})^{\dag}]\right]  \\
& =\frac{1}{d}\sum_{i}|i\rangle\!\langle i|_{A}\otimes\operatorname{Tr}%
_{A^{\prime}B^{\prime}E^{\prime}}\!\left[  U_{A^{\prime}B^{\prime}}^{ii}%
\sigma_{EA^{\prime}B^{\prime}E^{\prime}}(U_{A^{\prime}B^{\prime}}^{ii})^{\dag
}\right]  \\
& =\frac{1}{d}\sum_{i}|i\rangle\!\langle i|_{A}\otimes\sigma_{E}\\
& =\pi_{A}\otimes\sigma_{E}.
\end{align}
So we find that%
\begin{align}
\mathbf{D}(\gamma_{ABA^{\prime}B^{\prime}}\Vert\gamma_{AEA^{\prime}E^{\prime}%
})  \geq\mathbf{D}(\Phi_{AB}\Vert\pi_{A}\otimes\sigma_{E})
\geq I_{\mathbf{D}}(A;B)_{\Phi}.
\end{align}
Since this bound holds for an arbitrary extension $\gamma_{ABEA^{\prime}B^{\prime
}E^{\prime}}$\ of $\gamma_{ABA^{\prime}B^{\prime}}$, we conclude
\eqref{eq:lower-bnd-priv-state} after normalizing by $1/2$.

\section{Proof of Theorem~\ref{thm:overhead-bnd-e-max-private}}

\label{appx:overhead-bnd-e-max-private}

We first show that an arbitrary generalized unextendible entanglement $\mathbf{E}^{u}$,
which satisfies selective two-extendible monotonicity,
normalization, and subadditivity, serves as a lower bound on the overhead. Let us suppose that the selective two-extendible
operation $\Lambda$ outputs $\gamma^k_{ABA'B'}$ from $\rho_{AB}^{\ox n}$ with probability $p$.
Without loss of generality, we can assume that
\begin{align}
\Lambda(\rho_{AB}^{\ox n})=p\proj 0_{X_A}\ox\gamma^k_{ABA'B'} + (1-p)\proj 1_{X_A} \ox \tau,
\end{align}
where $\tau$ is some bipartite state and $X_A$ is a flag system in Alice's possession indicating whether the conversion is successful. Since 1-LOCC operations are allowed for free, it is possible for Alice to communicate $X_A$ to Bob for free. We
have
\begin{align}
  n\mathbf{E}^{u}(\rho_{AB})
\geq \mathbf{E}^{u}(\rho_{AB}^{\ox n})
\geq p\mathbf{E}^{u}(\gamma^k_{ABA'B'})
\geq pk,\label{eq:proof-overhead-1-ent}
\end{align}
where the first inequality follows from subadditivity, the second inequality follows from selective two-extendible
monotonicity, and the last inequality follows from Corollary~\ref{cor:lower-bnd-unext-alpha-private}.
Since $n$ and $p$ are arbitrary, we conclude that the
following bound holds for every integer $n \geq 1$ and $p\in(0,1]$:
\begin{align}
\frac{n}{p}\ge  \frac{k}{\mathbf{E}^{u}(\rho_{AB})}.
\end{align}
As the relative-entropy-induced unextendible entanglement $E^{u}(\rho_{AB})$ satisfies all of these required properties,
we conclude the lower bound in \eqref{eq:lower-bound-private}.

\section{Proof of Theorem~\ref{thm:exact-KD}}
\label{appx:exact-KD}

For a given bipartite state $\rho_{AB}$, suppose that there is a two-extendible channel $\Lambda_{AB\to
\hat{A}\hat{B}}$ that transforms $\rho_{AB}$ to a private state $\gamma^k_{\hat{A}\hat{B}A'B'}$ with $k$ private bits.
By the monotonicity of the min-unextendible entanglement, the following inequality holds
\begin{align}
E_{\min}^u(\rho_{AB})
\ge E_{\min}^u(\Lambda_{AB\to \hat{A}\hat{B}}(\rho_{AB}) )
= E_{\min}^u(\gamma^k_{\hat{A}\hat{B}A'B'})
\geq k,
 \label{eq:proof-exact-ED-1}
\end{align}
where the last inequality follows from Corollary~\ref{cor:lower-bnd-unext-alpha-private}.
Therefore, by optimizing over all two-extendible protocols,
it follows that the one-shot exact distillable entanglement of $\rho_{AB}$ is bounded as
\begin{align}
K_{\operatorname{2-EXT}}^{(1)}(\rho_{AB}) \le E_{\min}^u(\rho_{AB}).
\end{align}
Applying the same reasoning to the tensor-power state $\rho_{AB}^{\otimes n}$, we find that
\begin{align}
K_{\operatorname{2-EXT}}(\rho_{AB})
=
\liminf_{n\to \infty} \frac1n K_{\operatorname{2-EXT}}^{(1)}(\rho_{AB}^{\ox n})
\le  \liminf_{n\to \infty} \frac1n E_{\min}^u(\rho_{AB}^{\ox n}) = E_{\min}^u(\rho_{AB}),
\end{align}
where the final equality is a consequence of Proposition~\ref{prop:add-min-unext-ent}.

\section{Proof of Proposition~\ref{prop:unext-erased-state}}
\label{appx:unext-erased-state}

Note that the erased state $\rho^\varepsilon_{A'B}$ can be written in the following direct-sum form:
\begin{align}
  \rho^\varepsilon_{A'B} = (1-\varepsilon)\Phi^2_{AB} \oplus \varepsilon\pi_B
= \begin{bmatrix}
    (1-\varepsilon)\Phi^2_{AB} & 0 \\
    0 & \varepsilon\pi_B
\end{bmatrix}.
\end{align}
From this fact, one can see that each extension $\rho_{A'BB'}$ of $\rho^\varepsilon_{A'B}$ has the form
\begin{align}\label{eq:sigma_ABB}
    \rho_{A'BB'} \coloneqq
    \begin{bmatrix}
    (1-\varepsilon)\sigma_{ABB'} & 0 \\
    0 & \varepsilon\sigma_{BB'}
\end{bmatrix}
\end{align}
such that $\tr_{B'}\sigma_{ABB'}=\Phi^2_{AB}$ and $\tr_{B'}\sigma_{BB'} = \pi_B$.
Furthermore, any extension of $\Phi^2_{AB}$ necessarily has the form
$\sigma_{ABB'} \coloneqq  \Phi^2_{AB}\ox\tau_{B'}$, where $\tau_{B'}$ is a state on system $B'$. Thus
\begin{align}
  E^u(\rho^\varepsilon_{A'B})
& =  \min_{\rho_{A'BB'} \text{~st. Eq.~\eqref{eq:sigma_ABB}}}
      \frac{1}{2}D\left(\rho^\varepsilon_{A'B} \Vert \rho_{A'B'}\right) \\
& =  \min_{\tau_{B'}, \sigma_{B'}}
      \frac{1}{2}D\left(
        \begin{bmatrix}
          (1-\varepsilon)\Phi^2_{AB} & 0 \\
          0 & \varepsilon\pi_B
        \end{bmatrix} \middle\Vert
        \begin{bmatrix}
          (1-\varepsilon)\pi_A\ox\tau_{B'} & 0 \\
          0 & \varepsilon \sigma_{B'}
        \end{bmatrix}
      \right) \\
& =  (1-\varepsilon)\min_{\tau_{B'}}\frac{1}{2}D\left(\Phi^2_{AB}\Vert\pi_A\ox\tau_{B'}\right)
    + \varepsilon\inf_{\sigma_{B'}}\frac{1}{2}D\left(\pi_B\Vert\sigma_{B'}\right) \\
& =  (1-\varepsilon)\frac{1}{2}I(A;B)_\Phi \\
& = (1-\varepsilon),
\end{align}
where $I(A;B)_{\rho}$ is the quantum mutual information of $\rho_{AB}$.

\section{Proof of Proposition~\ref{prop:erasure-state-bounds}}
\label{appx:erasure-state-bounds}

Since $1/E_R(\rho)$ is a lower bound on the overhead of distillation using selective LOCC, 
it follows from the optimality
of erased state distillation under one-way LOCC that $1/(1-\varepsilon)
\geq 1/E_R(\rho_{A'B}^\varepsilon)$. This gives $E_R(\rho_{A'B}^\varepsilon)\geq1-\varepsilon$.

It then suffices to construct a feasible state to achieve the lower bound, and here we follow an approach similar to that from \cite[Proposition~11]{tomamichel2016strong}. Consider the following state
\begin{align}
    \sigma_{A'B} \coloneqq  \frac{1-\varepsilon}{2}\left(\proj{00} + \proj{11}\right)
                  + \varepsilon|e\rangle\!\langle e|\otimes\pi_B.
\end{align}
The state $\sigma_{A'B}$ is a separable state. It then holds that
\begin{align}
    E_R(\rho_{A'B}^\varepsilon)
& \leq  D(\rho_{A'B}^\varepsilon\Vert\sigma_{A'B}) \\
& =     D\left((1-\varepsilon)\Phi^2_{AB} \oplus \varepsilon\pi_B\middle\Vert
          \frac{1-\varepsilon}{2}\left(\proj{00} + \proj{11}\right)
          \oplus\varepsilon\pi_B\right) \\
& = (1-\varepsilon)D(\Phi^2_{AB}\Vert \left(\proj{00} + \proj{11}\right)/2) \\
& = 1-\varepsilon.
\end{align}
This concludes the proof.

\end{document}